\documentclass[10pt,journal]{IEEEtran}

\usepackage{amssymb}
\usepackage{amsmath}
\usepackage{cite}
\usepackage{url}
\usepackage{xcolor}
\usepackage{cite,graphicx,amsmath,amssymb}
\usepackage{fancyhdr}
\usepackage{mdwmath}
\usepackage{mdwtab}
\usepackage{caption}
\usepackage{amsthm}
\usepackage{setspace}
\usepackage{hyperref}
\usepackage{algorithm}
\usepackage{algorithmic}
\usepackage{multirow}
\usepackage{makecell}
\usepackage{mathtools}
\usepackage{subcaption}
\usepackage{bm}
\usepackage{tikz}
\usepackage{xurl}
\usepackage{stfloats}
\usepackage{mathrsfs}

\usepackage{booktabs}
\usepackage{multirow}
\usepackage{siunitx}
\usepackage{balance}

\hypersetup{colorlinks=true,
linkcolor=blue,
citecolor=blue,      
urlcolor=black,
}

\graphicspath{{./src/img/}}

\newtheorem{remark}{Remark}
\newtheorem{theorem}{Theorem}

\newtheorem{lemma}{Lemma}

\newtheorem{corollary}{Corollary}

\allowdisplaybreaks
\NewDocumentCommand{\multiubrace}{mmm}
 {
  \egreg_multiubrace:nnn {#1} {#2} {#3}
 }

\expandafter\def\expandafter\normalsize\expandafter{%
	\normalsize%
	\setlength\abovedisplayskip{4pt}%
	\setlength\belowdisplayskip{4pt}%
	\setlength\abovedisplayshortskip{2pt}%
	\setlength\belowdisplayshortskip{2pt}%
}
\begin{document}
\title{Generative Site-Specific Beamforming via Information-Maximizing Codebook} 
\author{{Cheng-Jie Zhao, Zhaolin Wang,~\IEEEmembership{Member, IEEE}, and Yuanwei Liu,~\IEEEmembership{Fellow, IEEE}}
	\vspace{-0.5cm}

\thanks{Cheng-Jie Zhao, Zhaolin Wang, and Yuanwei Liu are with Department of Electrical and Electronic Engineering, The University of Hong Kong, Hong Kong (e-mail: chengjie\_zhao@connect.hku.hk; zhaolin.wang@hku.hk; yuanwei@hku.hk).}
}
\maketitle

\begin{abstract}
	A novel generative site-specific beamforming (GenSSBF) framework is proposed, which integrates a site-information-maximizing (SIM) codebook with a conditional flow matching (CFM)-based beam generator. By this framework, the site-specific radio propagation environment is learned at the base station (BS), enabling the generation of high fidelity communication beams from coarse reference-signal-received-power (RSRP) feedback provided by user equipments (UEs). In the proposed design, a low-dimensional SIM probing codebook is first constructed by maximizing the mutual information between the RSRP feedback and the site-specific channel. This design not only reduces the initial beam sweeping overhead, but also enhances the amount of channel state information conveyed through UE feedback. By treating the RSRP feedback as a conditional prior, a CFM-based generative model is further developed to explicitly capture the uncertainty in beam generation. Specifically, a small set of UE-specific candidate beams is generated by inferring the learned generative model and sampling from the corresponding posterior distribution, after which the final data transmission beam is selected by the UE. Extensive simulation results demonstrate the effectiveness of both the proposed SIM codebook and the CFM-based beam generator. The proposed GenSSBF framework achieves beamforming performance nearly identical to maximum ratio transmission while requiring only eight probing beams and eight candidate beams.
\end{abstract}

\vspace{-0.3cm}
\section{Introduction}
Over the past decades, the rapid growth of mobile data traffic has continuously pushed cellular communication systems toward higher spectral efficiency and denser deployments. In this evolution, the spatial domain has emerged as a crucial resource, where beamforming enables directional transmission and reception through multi-antenna signal processing \cite{MIMO}. By concentrating signal energy toward intended user equipments (UEs) and mitigating multi-user interference, beamforming substantially enhances coverage, capacity, and energy efficiency. With large-scale antenna arrays and the migration toward millimeter-wave frequencies, beamforming has become a fundamental mechanism for ensuring reliable wireless communication in fifth-generation (5G) cellular networks \cite{mMIMO}. \\
\indent In practical 5G systems, beamforming is realized through standardized beam management and channel-dependent refinement mechanisms \cite{3GPPNR}. Due to highly directional transmission, reliable link establishment relies on initial beam sweeping and alignment, where the base station (BS) scans a predefined set of beams and the UE reports received signal quality, such as reference signal received power (RSRP), to identify suitable beam pairs \cite{IA}. Following this coarse alignment, beamforming can be further refined using up-link (UL) sounding reference signals (SRSs), enabling instantaneous channel state information (CSI) acquisition via UL–down-link (DL) reciprocity in time-division duplex systems \cite{Marzetta}, or through codebook-based feedback using precoding matrix indicators (PMIs) \cite{PMI}. While SRS-based beamforming offers high accuracy, its overhead scales rapidly with antenna dimensions and channel dynamics. In contrast, codebook-based approaches significantly reduce overhead but suffer from limited flexibility. These mechanisms reveal a fundamental trade-off between beamforming optimality and system overhead in current 5G frameworks. \\
\indent This trade-off originates from the universal design philosophy of conventional beamforming, which aims to operate reliably across diverse deployment scenarios. However, real-world propagation environments exhibit distinctive and persistent site-specific characteristics. Factors such as geometry, dominant scatterers, and blockage patterns impose strong structural constraints on channel realizations, effectively restricting them to a low-dimensional and highly structured subset. Exploiting this observation motivates site-specific beamforming (SSBF), which leverages prior knowledge of the propagation environment to guide beam design and alignment. \\
\indent SSBF departs from universal beamforming paradigms by explicitly incorporating site-dependent knowledge into the beamforming pipeline \cite{SSBFMAG1, SSBFMAG2}. By constraining channel realizations using long-term environmental information, SSBF reduces the reliance on exhaustive real-time CSI acquisition and enables high-performance beamforming with limited feedback. To maintain compatibility with existing 5G protocols, most SSBF schemes adopt a coarse-to-refinement framework similar to conventional systems. A probing codebook is first swept to obtain coarse measurements, followed by site-aware beam refinement based on the probing feedback. Within this framework, both probing and refinement stages can be designed in a site-specific manner. \\
\indent Existing SSBF studies can be broadly categorized by how site knowledge is exploited. One line of work focuses on site-specific probing codebook design, where probing beams are learned or optimized based on environmental characteristics \cite{SSBFCB1,SSBFCB2,SSBFCB3,SSBFCB4}. Another line of work moves beyond predefined codebooks and investigates joint probing and beam synthesis, enabling direct beam generation from continuous beam spaces using a small number of site-specific measurements \cite{SSBF1,E2ESSBF,CKM}. In these approaches, deep learning plays a key role due to the intrinsic complexity and strong structural dependence of site-aware beamforming. \\
\indent However, most existing SSBF frameworks rely on discriminative deep learning models, i.e., discriminative SSBF (DisSSBF), that map probing measurements directly to deterministic beamforming decisions \cite{SSBFCB1,SSBFCB2,SSBFCB3,SSBF1,E2ESSBF}. Common probing measurements such as RSRP provide only partial observations of the channel and inherently discard phase information and small-scale variations. As a result, a single probing observation may correspond to multiple distinct channel realizations. Discriminative models, which produce a single point estimate, are therefore unable to capture this intrinsic ambiguity, leading to limited robustness and generalization under model mismatch or unseen environments. \\
\indent In contrast to DisSSBF, generative models provide a principled framework for characterizing the conditional distribution of feasible channel realizations or beamformers given site-specific priors and limited measurements. By explicitly modeling uncertainty and capturing the structure of the underlying channel subspace, generative site-specific beamforming (GenSSBF) can produce multiple plausible beamforming solutions that are consistent with the same coarse observation. This capability allows GenSSBF to approach the performance of SRS-based beamforming using only lightweight refinement procedures, while requiring significantly lower signaling overhead than conventional beam alignment mechanisms. A comprehensive analysis of the advantages of GenSSBF over DisSSBF is provided in \cite{SSBFMAG2}. More recently, diffusion models \cite{DDPM} have been introduced as an effective realization of GenSSBF in \cite{SSBF2}, where substantial performance gains over DisSSBF approaches have been demonstrated through simulations. \\
\indent Despite their performance advantages, diffusion-based GenSSBF methods suffer from inherent limitations in practical deployment. In particular, diffusion models rely on long iterative sampling processes, which introduce non-negligible inference latency and computational overhead, rendering real-time beamforming challenging. To address this limitation, flow matching (FM) \cite{fm1} has recently emerged as an efficient generative modeling framework for learning complex continuous distributions. Instead of iterative denoising, FM directly learns a vector field that transports samples from a base distribution to the target distribution, enabling high-quality sample generation within only a few integration steps. When conditioned on RSRP vectors, conditional flow matching (CFM) can efficiently learn the posterior distribution of optimal beamformers and rapidly generate candidate beams for refinement. \\
\indent In addition to efficient beamformer generation, effective probing codebook design is a critical component of SSBF, since channel information is implicitly acquired through RSRP measurements of transmitted probing beams. However, commonly adopted discrete Fourier transform (DFT)-based codebooks are not optimized for information extraction in site-specific environments and may lead to inefficient or redundant probing, as will be shown later. Motivated by this observation, beyond introducing CFM as an efficient generative beam synthesis model, we further develop a information-theoretical approach for site-specific probing codebook design. \\
\indent The main contributions of this paper are summarized as follows:
\begin{itemize}
	\vspace{-0.3cm}
	\item We propose a novel GenSSBF framework for beam management. Within each channel coherence block, the proposed framework follows a three-stage coarse-to-refinement procedure consisting of channel probing, beam refinement, and beam locking. Specifically, the BS probes the channel using an optimized SIM codebook and collects RSRP measurements from UE feedback, based on which a CFM-based model generates a small set of candidate beams for refinement. The strongest beam is then selected for data transmission. The proposed framework is compatible with standard 5G beam management procedures and achieves near-optimal beamforming performance with significantly reduced probing overhead.
	\item We develop an information-theoretic SIM codebook design for low-overhead channel probing. By interpreting probing beams as measurement operators of the underlying site-specific channel, the SIM codebook is constructed by maximizing the mutual information between the RSRP feedback and the channel, so that UE feedback conveys more informative channel knowledge while the initial beam sweeping overhead is substantially reduced.
	\item We propose a CFM-based beam generator that treats the RSRP feedback as a conditional prior and explicitly accounts for the uncertainty in beam generation. By inferring the learned generative model and sampling from the corresponding posterior distribution, the BS produces a small set of UE-specific candidate beams to resolve the remaining uncertainty. Feature-wise linear modulation (FiLM) is adopted for condition embedding, enabling the generator to effectively interpret RSRP measurements and guide the generative process, thereby improving robustness under partial and noisy measurements.
	\item Extensive simulation results validate the effectiveness of the proposed framework. In particular, the designed probing codebook is shown to be more informative and robust than conventional DFT-based codebook. The FiLM-aided CFM model achieves near-optimal performance with short inference time. Eventually, the proposed framework achieves performance nearly identical to maximum ratio transmission while requiring only eight probing beams and eight candidate beams.
\end{itemize}
\vspace{-0.3cm}
The remainder of this paper is structured as follows. Section~\ref{sec2} presents the proposed GenSSBF framework, develops the corresponding system model, and formulates the associated optimization problem. Section~\ref{sec3} details the design principles and algorithm for the SIM codebook. Section~\ref{sec4} introduces the CFM-based beam generator. Numerical results that evaluate the performance of the proposed GenSSBF framework under various system configurations are provided in Section~\ref{sec5}. Finally, Section~\ref{sec6} concludes the paper. \\
\indent \emph{Notations:} Scalars, vectors/matrices, and sets are denoted by regular, lowercase/uppercase boldface, and calligraphic letters, respectively. the real and complex number fields are denoted by $\mathbb{R}$ and $\mathbb{C}$, respectively. The inverse, transpose, and conjugate transpose are represented by $(\cdot)^{-1}$, $(\cdot)^T$, and $(\cdot)^H$, respectively. The absolute value and Euclidean norm are indicated by $|\cdot|$ and $\|\cdot\|$, respectively. The expectation operator is denoted by $\mathbb{E}\left[ {\cdot}\right]$. An N-dimensional all-one column vector is denoted by $\mathbf{1}_N$. An identity matrix of size $N \times N$ is denoted by $\mathbf{I}_N$. $[x]_+$ denotes the positive part of $x$, i.e., $\max (x,0)$.
\vspace{-0.3cm}
\section{System Model and Problem Formulation} \label{sec2}
In this section, we present the proposed GenSSBF framework. We first describe the system setup and the overall beam management procedure, and then establish the corresponding mathematical models.
\begin{figure*}[tb]
	\centering
	\includegraphics[scale=0.42]{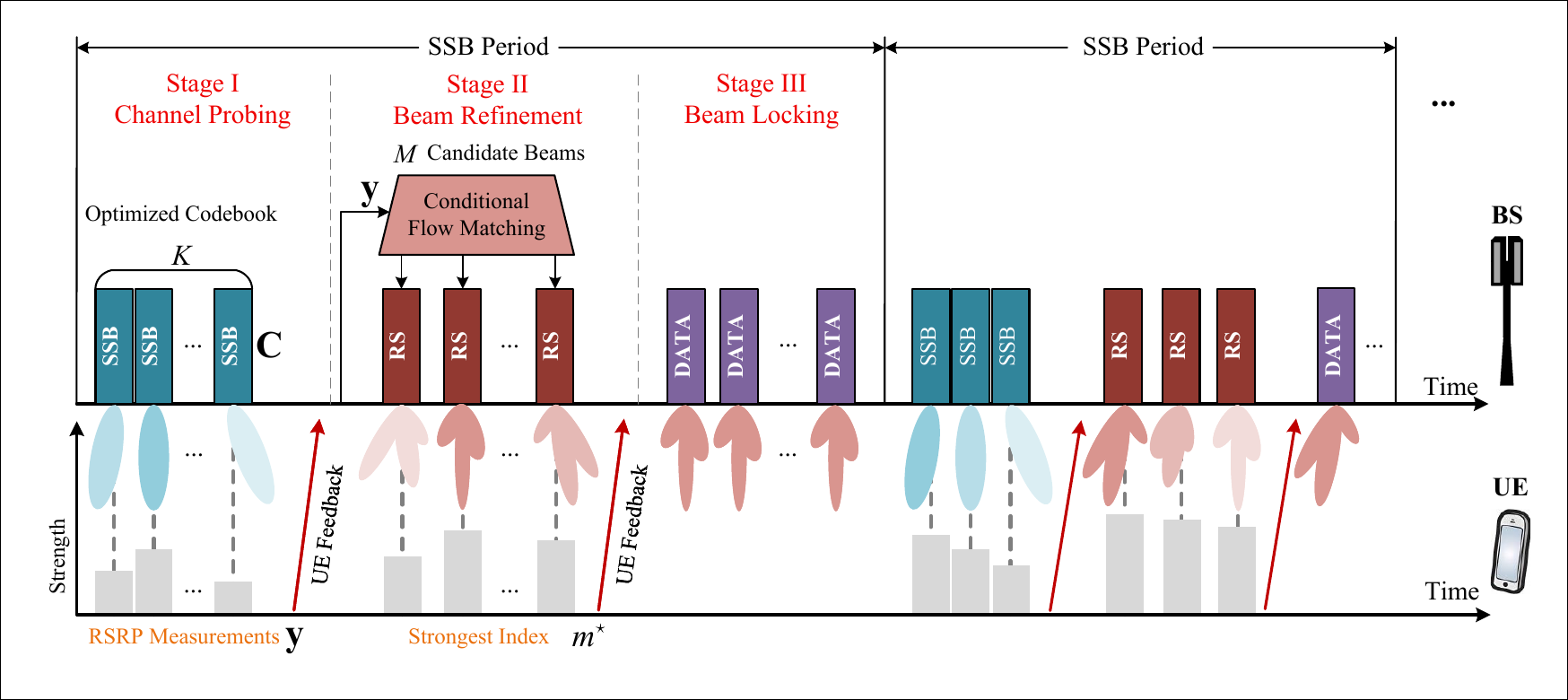}
	\caption{Illustration for the proposed GenSSBF framework}
	\label{protocol}
	\vspace{-0.5cm}
\end{figure*} 
\vspace{-0.5cm}
\subsection{Beam Management Procedure}
We consider a BS equipped with $N_t$ antennas communicating with a single-antenna UE, which is randomly located within a predefined two-dimensional (2D) service area.
The BS is assumed to employ a single radio frequency (RF) chain and performs analog beamforming only. 
The overall beam management procedure between the BS and the UE is illustrated in Fig .\ref{protocol}. We assume a block-fading channel model, where the channel remains approximately constant within each beam management interval, which is referred to as an SSB period in Fig.~\ref{protocol}. Within each SSB period, the beam management procedure consists of three stages: channel probing, beam refinement, and beam locking. 
\begin{itemize}
	\item \textbf{Channel Probing:} The BS scans a predefined codebook of beams using SSB transmissions. Based on the received signals, the UE performs initial access (IA) and synchronization with the BS \cite{3GPP_R1-166088}, and subsequently reports the measured RSRP for each probing beam back to the BS. Unlike conventional 5G new radio (NR) systems, which typically utilize a DFT-based codebook, the probing codebook in the proposed framework is optimized using site-specific data, as detailed in Section~\ref{sec3}.
	\item \textbf{Beam Refinement:} The BS treats the received RSRP feedback as conditional information and applies the CFM model, presented in Section~\ref{sec4}, to generate a small set of candidate beams. These candidate beams are transmitted to the UE using reference signals (RSs). The UE then measures the received power of each beam and feeds back the index of the strongest one.
	\item \textbf{Beam Locking:} The beam selected in the second stage is used for data transmission.
\end{itemize} 
This procedure is repeated in every channel coherence block. 
\vspace{-0.3cm}
\subsection{Signal Model}
Let $\mathbf{s}_{\rm SSB} \in \mathbb{C}^{L_s \times 1}$ denote the transmitted SSB signal, where $L_s$ is the signal length. The SSB signal is assumed to satisfy $\mathbb{E}\left[{{\bf{s}}_{{\rm{SSB}}}} {{\bf{s}}_{{\rm{SSB}}}^H} \right] = {P_t}\mathbf{I}_{L_s}$, where $P_t$ denotes the transmit power. The received SSB signal at the UE can be expressed as
\begin{equation}
	{{\bf{r}}_{{\rm{SSB}}}} = \mathbf{h}^H\left( \mathcal{E}_s \right)\mathbf{w}{{\bf{s}}_{{\rm{SSB}}}} + {\bf{n}}.
\end{equation}
Here, $\mathbf{h} \left( \mathcal{E}_s \right) \in \mathbb{C}^{N_t \times 1}$ denotes the DL channel vector between the BS and the UE, which is determined by the site-specific propagation environment $\mathcal{E}_s$. Vector $\mathbf{w} \in \mathbb{C}^{{N_t}\times 1}$ represents the unit-modulus analog beamformer adopted at the BS. The noise vector $\mathbf{n} \in \mathbb{C}^{{L_s}\times 1}$ is modeled as circularly symmetric complex Gaussian noise with zero mean and covariance $\sigma_n^2 \mathbf{I}_{L_s}$, i.e., $\mathbf{n} \sim \mathcal{CN}(\mathbf{0}_{L_s}, \sigma_n^2 \mathbf{I}_{L_s})$, and is assumed to be temporally independent of $\mathbf{s}_{\rm SSB}$.  \\
\indent Within each channel coherence block, the DL channel between the BS and the UE is modeled as a superposition of $L$ propagation paths \cite{MIMO}
\begin{equation} \label{h}
	\mathbf{h} = \sum\limits_{l = 1}^{L} {{\alpha _l}{{\bf{a}}}\left( {{\varphi _l}} \right)}.
\end{equation} 
where $L$ denotes the number of resolvable propagation paths determined by the site-specific propagation environment $\mathcal{E}_s$ and $\varphi_l$ and $\alpha_l = \sqrt{\beta_l}\exp\!\left( j\psi_l \right)$ denote the angle of departure (AoD) and complex gain associated with the $l$-th path, respectively. More particularly, $\psi_l$ represents the phase shift induced by propagation and reflection effects and is commonly modeled as being uniformly distributed over $[-\pi,\pi)$. $\beta_l$ is the large-scale power attenuation modeled as \cite{MIMO}
\begin{equation}
	\beta_l = {C_0}{d_l^{ - {\alpha _{{\rm{PL}}}}}}S_l,
\end{equation}
where $C_0$ is a reference gain constant, $d_l$ is the propagation distance of the $l$-th path, and $\alpha_{\rm PL}$ denotes the path loss exponent, which typically takes values between $2$ and $4$ in cellular systems \cite{MIMO}. The term $S_l$ represents shadow fading caused by obstacles such as buildings, vegetation, or vehicles, and is commonly modeled as a log-normal random variable, i.e., $S_l \sim \mathcal{LN}(0,\sigma_{\text{sh},l}^2)$. As a consequence, the overall power attenuation also follows a log-normal distribution given by ${\beta _l} \sim \mathcal{LN}\left( {{C_0}{d_l^{ - {\alpha _{{\rm{PL}}}}}}, \sigma_{\text{sh},l}^2} \right)$. Finally, $\mathbf{a}(\varphi) \in \mathbb{C}^{{N_t}\times 1}$ denotes the transmit steering vector of the BS corresponding to AoD $\varphi$.
\vspace{-0.45cm}
\subsection{RSRP Model}
\vspace{-0.1cm}
The RSRP associated with a probing beam is obtained from the corresponding RSs and is used to quantify the received signal strength along that beam direction. Taking the SSB signal as an example, the RSRP is estimated by averaging the received power over a finite number of SSB signal elements, which can be written as
\begin{align}
	&\hat{P}_\mathbf{w} \triangleq \frac{1}{{{L_s}}}\sum\limits_{t = 1}^{{L_s}} {{{\left| {{r_{{\rm{SSB}},t}}} \right|}^2}}  \notag \\
	&= \frac{1}{{{L_s}}}\left( {{{\left| {\mathbf{h}^H\mathbf{w}} \right|}^2}\left\| {{{\bf{s}}_{{\rm{SSB}}}}} \right\|^2 + \left\| {\bf{n}} \right\|^2 + 2{\mathop{\rm Re}\nolimits} \left\{ {\mathbf{h}^H\mathbf{w}{\bf{s}}_{{\rm{SSB}}}^H{\bf{n}}} \right\}} \right),
\end{align}
where $r_{{\rm SSB},t}$ denotes the $t$-th element of vector ${{\bf{r}}_{{\rm{SSB}}}}$, i.e., the $t$-th received symbol on the RS resource elements of the SSB. \\
\indent For a given channel realization $\mathbf{h}$ and beamformer $\mathbf{w}$, the RSRP estimate $\hat{P}_\mathbf{w}$ can be approximated as a Gaussian random variable when $L_s$ is sufficiently large, according to the central limit theorem \cite{HDP}. Taking expectation with respect to the transmitted SSB signal and the thermal noise, the mean and variance of $\hat{P}_\mathbf{w}$ are given by
\begin{subequations}
	\begin{align}
		{\mu _p}  &= {\mathbb{E}_{{\bf{s_{\rm{SSB}}},\mathbf{n}}}}\left[ {\hat{P}_\mathbf{w}} \right] = {P_t}{\left| {\mathbf{h}^H\mathbf{w}} \right|^2} + {\sigma_n ^2}, \\
		\sigma _p^2 &= \frac{1}{{{L_s^2}}}\left( {\sigma _n^4 + 2\sigma _n^2{P_t}{\left| {\mathbf{h}^H\mathbf{w}} \right|^2}} \right).
	\end{align}
\end{subequations}
Accordingly, the RSRP estimate can be expressed as $\hat{P}_\mathbf{w}=\mu _p+n_p$, where $n_p\sim\mathcal{N}(0,\sigma_p^2)$ captures the uncertainty introduced by finite-length time averaging. In addition to thermal noise, the RSRP estimate is also affected by large-scale shadowing embedded in the channel $\mathbf{h}$. To facilitate tractable analysis, we assume that all propagation paths share a common large-scale shadowing coefficient, since the dominant blockage and attenuation are mainly determined by the same BS–UE geometry and surrounding obstacles within a local area. Specifically, we assume $S_l=S$ and $\sigma_{\text{sh},l}^2=\sigma_\text{sh}^2$ for all $l$. Under this assumption, the effective channel gain can be expressed as
\begin{equation}
	{\left| {{{\bf{h}}^H}{\bf{w}}} \right|^2} = S\left| {\sum\limits_{l = 1}^L {\sqrt {{C_0}d_l^{ - {\alpha _{{\rm{PL}}}}}} {e^{ - j{\psi _l}}}{{\bf{a}}^H}\left( {{\varphi _l}} \right){\bf{w}}} } \right| = SG_{\mathbf{w}},
\end{equation}
where $G_{\mathbf{w}}$ denotes the deterministic array gain associated with beamformer $\mathbf{w}$, given the site geometry and beamforming direction. \\
\indent As a result, the RSRP estimate consists of a log-normal shadowing component $S G_w + \sigma_n^2$ and an additive measurement noise component $n_p$. It is therefore convenient to analyze the RSRP in the decibel (dB) domain. Define
\begin{align} \label{logrsrp}
	y_\mathbf{w} &= 10{\log _{10}}\hat{P}_\mathbf{w} = 10{\log _{10}}\left( {{P_t}S{G_{\bf{w}}} + \sigma _n^2 + {n_p}} \right) \notag \\
	&\overset{(a)}{\approx} 10{\log _{10}}{P_t} + 10{\log _{10}}{G_{\bf{w}}} + 10{\log _{10}}S + \frac{10}{\ln10} \frac{{\sigma _n^2+{n_p}}}{{{{P_t}S{G_{\bf{w}}}}}} \notag \\
	&= {P_t^{{\rm{dB}}}} + G_{\bf{w}}^{{\rm{dB}}} + {n_y} = y_\mathbf{w}^0 + n_y,
\end{align}
where ${P_t^{{\rm{dB}}}}$ and $G_{\bf{w}}^{{\rm{dB}}}$ denote the dB-formed transmit power and array gain, respectively. Hence, $y_\mathbf{w}^0={P_t^{{\rm{dB}}}} + G_{\bf{w}}^{{\rm{dB}}}$ denotes the true RSRP value in the dB domain. In step (a), the approximation $\log(a+b)\approx\log(a)+b/(a\ln10)$ is applied, which is valid when $|b/a|\ll1$, corresponding to moderate or high signal-to-noise ratio (SNR) regimes. Under this approximation, the effective noise term $n_y$ approximately follows a Gaussian distribution with mean and variance
\begin{subequations}
	\begin{align}
		\mu_y &\approx \frac{10}{\ln10}{\frac{{\sigma _n^2}}{{{P_t}{G_{\bf{w}}}}}}, \\
		\sigma_y^2 &\approx \sigma _\text{sh}^2 + (\frac{10}{\ln10})^2\frac{\sigma _n^4 + 2\sigma _n^2{P_t}{G_{\bf{w}}}}{{{P_t^2}{G_{\bf{w}}^2}L_s^2}}. \label{8b}
	\end{align}
\end{subequations}
For analytical simplicity, the impact of shadowing on the mean bias and variance terms is neglected, as its contribution is assumed to be minor relative to the overall received signal strength in the RSRP measurement. Consequently, the measured RSRP in the dB domain can be approximated as 
\begin{equation}
	y_\mathbf{w} \sim \mathcal{N}(y_\mathbf{w}^0+\mu_y,\sigma_y^2).
\end{equation} 
From the expression of $\sigma_y^2$, it can be observed that the uncertainty of the RSRP measurement is mainly governed by the shadowing variance $\sigma_\text{sh}^2$ and the thermal noise power $\sigma_n^2$. Moreover, for beamforming directions with larger array gains $G_\mathbf{w}$, both $\mu_y$ and $\sigma_y^2$ become smaller, indicating that the measured RSRP is closer to its ground-truth value. This is because noise perturbations in the logarithmic domain are relative to the received signal strength, and stronger received signals lead to smaller relative distortions in the dB scale. \\
\indent Suppose that the designed codebook contains $K$ probing beams, the codebook can be represented by a matrix $\mathbf{C}=[\mathbf{c}_1,\cdots,\mathbf{c}_K] \in \mathcal{C}^{N_t \times K}$, where $\mathcal{C}$ denotes the complex unit circle. Due to the unit-modulus constraint imposed by analog beamforming architectures, $\mathbf{C}$ lies on a Riemannian manifold. After sweeping all $K$ beams in the codebook, the UE feeds back $K$ measured RSRP values to the BS in dB scale. Each element ${y_{\mathbf{c}_k}}$ of the RSRP vector ${\bf{y}} = {\left[ {{y_{\mathbf{c}_1}}, \cdots ,{y_{\mathbf{c}_K}}} \right]^T} \in \mathbb{R}^K$ corresponds to the measurement obtained using the $k$-th probing beam. Following the model developed in Eq.~(\ref{logrsrp},) the RSRP vector ${\bf{y}}$ is modeled as a Gaussian random vector given by
\begin{equation} \label{measurement}
	{\bf{y}} = {\bf{y}}_0 + {\bf{n}}_y = \mathcal{M}_\mathbf{C}(\mathbf{h}(\mathcal{E}_s)) + {\bf{n}}_y, 
\end{equation}
where ${\bf{y}}_0$ denotes the noise-free RSRP vector, while ${\bf{n}}_y$ represents perturbations caused by thermal noise and shadowing effects. The $k$-th element of the mean vector $\boldsymbol{\mu}_y$ of ${\bf{y}}$ is given by ${P_t^{{\rm{dB}}}} + G_{\mathbf{c}_k}^{{\rm{dB}}} + \mu_{y,\mathbf{c}_k}$, and the covariance matrix is $\mathbf{\Sigma}_y=\text{diag}\left\lbrace \sigma_{y,\mathbf{c}_1}^2, \cdots, \sigma_{y,\mathbf{c}_K}^2 \right\rbrace$. From an alternative perspective, the probing codebook can be interpreted as a measurement operator applied to the underlying channel. Under this interpretation, the RSRP is viewed as a vector generated by the nonlinear measurement mapping $\mathcal{M}_\mathbf{C}$, which is fully determined by the probing codebook $\mathbf{C}$.
\vspace{-0.3cm}
\subsection{Problem Formulation}
Based on the developed mathematical models, the beam management problem is formulated as follows:
\begin{subequations}\label{p1}
	\begin{align} 
		\left( {\rm{P1}}\right)  \quad \quad  \mathop {\max }\limits_{{\bf{C}},f} \quad & u\left( {{\bf{h}}(\mathcal{E}_s),f\left( {{\bf{y}}} \right)} \right)  \\
		\label{1a} \mathrm{s.t.} \ \  & (\ref{measurement}), \mathbf{C} \in \mathcal{S}_\mathbf{C}, f\left( {{\bf{y}}} \right) \in \mathcal{S}_\mathbf{w}.
	\end{align}
\end{subequations}
In (P1), $f$ denotes a generator that maps the RSRP measurement vector ${\bf y}$ to a beamformer. The probing codebook $\mathbf C$ and the mapping $f$ are jointly optimized to maximize a beamforming utility function $u({\bf h},{\bf w}): \mathbb{C}^{N_t \times 1} \times \mathbb{C}^{N_t \times 1} \rightarrow \mathbb{R}$, which characterizes performance metrics such as received SNR. The feasible sets $\mathcal S_{\mathbf C}$ and $\mathcal S_{\mathbf w}$ capture practical physical constraints on the probing codebook and the generated beamformer, including power and unit-modulus constraints.  \\
\indent Due to the nonlinear measurement model, the lack of instantaneous channel knowledge, and the functional optimization over $f$, problem (P1) is highly non-convex and intractable to solve directly. To overcome this difficulty, we decompose (P1) into two sub-problems, i.e., probing codebook design and beam generation, which correspond to the channel probing and beam refinement stages of the proposed GenSSBF framework.

\section{Site-Information-Maximizing Codebook} \label{sec3}
This section focus on the probing codebook design sub-problem, which is formulated as
\begin{subequations}\label{p2}
	\begin{align} 
		\left( {\rm{P2}}\right) \quad \quad  \  \mathop {\max }\limits_{{\bf{C}}} \quad & u\left( {{\bf{h}}(\mathcal{E}_s),f\left( {{\bf{y}}} \right)} \right)  \\
		\label{2a} \mathrm{s.t.} \ \  & (\ref{measurement}), \mathbf{C} \in \mathcal{S}_\mathbf{C} .
	\end{align}
\end{subequations}
It is worth noting that the probing codebook $\mathbf C$ appears in the objective of (P2) only implicitly. Specifically, any beamforming decision is made based on the probing measurements, i.e., the RSRP vector $\mathbf{y}$ obtained from the probing codebook $\mathbf C$. As a result, the dependence of the beamforming utility on the underlying channel is entirely mediated through the measurement vector $\mathbf{y}$, and hence through the choice of $\mathbf C$. In this sense, the performance of problem (P2) is fundamentally constrained by how much channel information is preserved in the probing measurements.

\vspace{-0.5cm}
\subsection{Information-Theoretical Design Principles}
\vspace{-0.1cm}
\subsubsection{Mutual Information Maximization}
The above observation naturally motivates an information-theoretic design principle for the probing codebook. This principle is further supported by the fact that, given perfect CSI, the optimal analog beamforming vector is typically a deterministic function of the channel. For example, in a single-user, single-stream, narrow-band setting, the analog beamformer that maximizes the received SNR is given by the maximum-ratio transmission (MRT) solution, i.e., $\mathbf{w}_{\text{MRT}}^\star=\exp(j\angle \mathbf{h})$. \\
\indent Consequently, we adopt mutual information as a principled metric to quantify the informativeness of the probing measurements. Specifically, the mutual information between the RSRP vector $\mathbf{y}$ and the channel $\mathbf{h}$, denoted by, $I(\mathbf{y};\mathbf{h})$, characterizes how much channel information is retained in the probing stage. Accordingly, the probing codebook should be designed to maximize this mutual information, i.e., $\mathbf{C} = \arg \mathop {\max }\limits_\mathbf{C} I\left( {{\bf{y}};{\bf{h}}} \right)$. By definition, the mutual information is given by
\begin{equation}
	I\left( {{\bf{y}};{\bf{h}}} \right) = h\left( {\bf{y}} \right) - h\left( {{\bf{y}}\left| {\bf{h}} \right.} \right),
\end{equation}
where $h(\cdot)$ denotes the differential entropy. According to the measurement model in (\ref{measurement}), the conditional distribution $p(\mathbf y|\mathbf h)$ is Gaussian with covariance matrix $\mathbf{\Sigma}_y(\mathbf{h},\mathbf{C})$. Hence,
\begin{align}
	h(\mathbf y|\mathbf h)
	&= \frac{1}{2}\mathbb E_{\mathbf h}\!\left[\log\!\Big((2\pi e)^K\det \mathbf{\Sigma}_y(\mathbf{h},\mathbf{C})\Big)\right] \\
	&= \frac{1}{2}\mathbb E_{\mathbf h}\!\left[\sum\limits_{k = 1}^K {\frac{1}{2}\log \left( {2\pi e{\sigma_{y, \mathbf{c}_k}^2}}(\mathbf{h}) \right)}\right].
\end{align}
Since the marginal distribution of $\mathbf y$ is generally unknown due to the nonlinear
measurement mapping, we upper bound $h(\mathbf y)$ using the maximum-entropy property of the Gaussian distribution:
\begin{equation}
	h(\mathbf y)\le \frac12\log\!\Big((2\pi e)^K\det \mathbf{R}_y\Big),
\end{equation}
where $\mathbf{R}_y=\mathbb E[\mathbf y\mathbf y^T]-\mathbb E[\mathbf y]\mathbb E[\mathbf y]^T$.
Therefore,
\begin{align} \label{mibound}
	I\left( {{\bf{y}};{\bf{h}}} \right) &= h\left( {\bf{y}} \right) - \frac{1}{2}\mathbb E_{\mathbf h}\!\left[\sum\limits_{k = 1}^K {\frac{1}{2}\log \left( {2\pi e{\sigma_{y, \mathbf{c}_k}^2}}(\mathbf{h}) \right)}\right] \notag \\
	&\le \frac{1}{2}\log \det {\mathbf{\mathbf{R}} _y} - \frac{1}{2}\mathbb E_{\mathbf h}\!\left[\sum\limits_{k = 1}^K {\frac{1}{2}\log {\sigma_{y, \mathbf{c}_k}^2}(\mathbf{h}) }\right] \notag \\
	&\overset{(b)}{\approx} \frac{1}{2}\log \det {{\bf{R}}_y} - I_c.
\end{align}
Step (b) assumes that $\frac{1}{2}\mathbb E_{\mathbf h}\!\left[\sum\nolimits_{k = 1}^K {\frac{1}{2}\log {\sigma_{y, \mathbf{c}_k}^2}(\mathbf{h}) }\right]$ is approximated by a constant $I_c$. The reason is that the uncertainty of $\mathbf y$ given $\mathbf h$ is dominated by the site-specific propagation environment. The term $\frac{1}{2}\mathbb E_{\mathbf h}\!\left[\sum\nolimits_{k = 1}^K {\frac{1}{2}\log {\sigma_{y, \mathbf{c}_k}^2}(\mathbf{h}) }\right]$ thus varies weakly with the codebook. Maximizing the mutual information between $\mathbf{y}$ and $\mathbf{h}$ then reduces to maximizing $\log \det {{\bf{R}}_y}$. \\
\indent This objective also admits an intuitive geometric interpretation. The covariance matrix $\mathbf{R}_y$, defines an ellipsoidal confidence region in the measurement space ${\mathcal{E}_y}\left( {{c_e}} \right) = \left\{ {{\bf{y}}:{{\left( {{\bf{y}} - {\boldsymbol{\mu} _y}} \right)}^T}{\bf{R}}_y^{ - 1}\left( {{\bf{y}} - {\boldsymbol{\mu} _y}} \right) \le {c_e}} \right\}$, whose volume is proportional to $\det {{\bf{R}}_y}$. Maximizing this determinant enlarges the volume of the measurement ellipsoid, encouraging the probing beams to generate diverse and complementary RSRP measurements. \\
\indent Further insight can be obtained from an estimation-theoretic perspective by viewing the probing stage as a channel estimation problem, where the channel vector $\mathbf{h}$ is inferred from the RSRP measurements $\mathbf{y}$ generated by the probing codebook. From this viewpoint, the design of the probing codebook directly determines the measurement operator and hence the quality of channel inference. By locally approximating the nonlinear measurement mapping around a nominal channel realization $\mathbf{h}_0$, we have
\begin{equation}
	\mathbf{y} \approx \mathcal{M}_\mathbf{C}(\mathbf{h}_0)+{\bf{J}}_\mathbf{C}\left( {{{\bf{h}}_0}} \right)\left( {{\bf{h}} - {{\bf{h}}_0}} \right)+\mathbf{n}_y,	
\end{equation}
where ${\bf{J}}_\mathbf{C} = \frac{{\partial {\mathcal{M}_C}}}{{\partial {\bf{h}}}} \in {\mathbb{R}^{K \times 2{N_t}}}$ denotes the real-valued augmented Jacobian matrix, which is fully determined by the probing codebook. Under this local linear model, the Fisher information matrix (FIM) for estimating $\mathbf{h}$ is given by $\mathbf{F}(\mathbf{h}_0)= \mathbf{J}_\mathbf{C}^T\mathbf{\Sigma}_y^{-1}\mathbf{J}_\mathbf{C}$. The covariance matrix of $\mathbf{y}$ can thus be expressed as
\begin{equation}
	{{\bf{R}}_y} = {\bf{J}}_\mathbf{C}{{\bf{R}}_h}{{\bf{J}}_\mathbf{C}^T} + {\mathbf{\Sigma} _y},
\end{equation}
where ${{\bf{R}}_h}$ denotes the channel covariance matrix and is independent of the probing codebook. Applying the matrix determinant lemma yields
\begin{equation}
	\log \det {{\bf{R}}_y} = \log \det {\mathbf{\Sigma} _y} + \log \det \left( {{{\bf{I}}_{2N_t}} + {{\bf{R}}_h}{\bf{F}}} \right).
\end{equation} 
Since $\mathbf{\Sigma}_y$ and ${{\bf{R}}_h}$ are fixed, maximizing $\det {{\bf{R}}_y}$ is equivalent to maximizing $\det {\bf{F}}$, which corresponds to minimizing the volume of the Cramér–Rao lower bound (CRLB) ellipsoid, thereby enhancing robustness to Gaussian perturbations. Such criterion is also known as D-optimality in experiment design \cite{Doptimal}. 

\subsubsection{Orthogonality Construction}
Maximizing $\log \det {{\bf{R}}_y}$ is theoretically optimal from both information-theoretic and estimation-theoretic perspectives. However, this objective alone does not explicitly regulate certain structural properties that are desirable for a valid probing codebook. In particular, without additional constraints, the optimization may converge to stationary points where the increase in $\log \det {{\bf{R}}_y}$ is dominated by a few highly informative beams, while the remaining beams become strongly correlated and contribute marginally to the overall information gain. Such beam correlations lead to redundant measurements and may result in locally ill-conditioned information matrices, which is undesirable for robust probing and subsequent inference. \\
\indent To mitigate this issue, the probing beams should be designed to exhibit low mutual correlation, such that each beam contributes complementary information. To this end, we introduce an explicit orthogonality constraint to regularize the optimization of the probing codebook $\mathbf{C}$. Specifically, the correlation among different beams, corresponding to the columns of $\mathbf{C}$, is characterized by the Gram matrix $\mathbf{G}_\mathbf{C} \in \mathbb{C}^{K \times K}$, whose $(i,j)$-th element is defined as ${\textstyle{{{\bf{c}}_i^H{{\bf{c}}_j}} \over {\left\| {{{\bf{c}}_i}} \right\|\left\| {{{\bf{c}}_j}} \right\|}}}$. The off-diagonal elements of $\mathbf{G}_\mathbf{C}$ therefore quantify the inter-beam correlations within the probing codebook. For an orthogonal set of beams, the Gram matrix reduces to the identity matrix. Accordingly, we impose the following orthogonality condition:
\begin{equation} \label{ortho}
	\mathbf{G}_\mathbf{C} = \mathbf{I}_K.
\end{equation}
By encouraging the probing codebook to be orthogonal, the resulting beams are driven to generate statistically independent and non-redundant measurements. This constraint complements the log-determinant objective by preventing correlation collapse among beams and promoting well-conditioned information matrices, thereby enhancing the robustness and effectiveness of the probing stage.

\subsubsection{Coverage Guarantee}
Beyond maximizing the overall information content of the probing measurements, another critical concern is the spatial coverage of the designed probing codebook. In particular, it is necessary to explicitly prevent degenerate probing patterns in which all beams produce weak responses at certain spatial locations. In such scenarios, the resulting RSRP measurements are dominated by noise, leading to unreliable measurements and a poorly conditioned FIM, even if the codebook performs well on average. \\
\indent To address this issue, we introduce a coverage-guaranteeing constraint that ensures each spatial sample is sufficiently excited by at least one probing beam. This constraint enforces a minimum signal strength level across the region of interest, thereby preventing coverage holes in the probing stage. Specifically, the coverage requirement is expressed as
\begin{equation} \label{coverage}
	\mathop {\max }\limits_k {y_k} \ge {y_{{\rm{th}}}}, \text{for all users},
\end{equation}
where $y_{th}$ denotes a predefined threshold that guarantees adequate received signal strength. By enforcing this constraint, the probing codebook is encouraged to provide reliable measurements for all spatial locations under consideration. This coverage guarantee complements the mutual information maximization and orthogonality constraints by safeguarding against noise-dominated observations and ensuring that the probing stage yields well-conditioned and informative measurements throughout the entire service area.
\vspace{-0.2cm}
\subsection{Probing Codebook Optimization}
Based on the above information-theoretical design principles, we formulate the following optimization problem for probing codebook design:
\begin{subequations} \label{1}
	\begin{align} 
		\left( {\rm{P3}}\right)  \  \mathop {\max }\limits_{{\bf{C}}} \quad &\log \det \mathbf{R}_y  \\
		\label{3a} \mathrm{s.t.} \ \  & (\ref{ortho}), (\ref{coverage}), \left| {{{\bf{C}}_{i,j}}} \right| = 1, \forall i,j.
	\end{align}
\end{subequations}
Solving (P3) directly is also intractable. First, $\log \det \mathbf{R}_y$ is concave for the positive semi-definite $\mathbf{R}_y$, but $\mathbf{C}$ is implicitly involved in the objective function. Additionally, the expression for $\mathbf{y}$ in (\ref{measurement}) is highly non-linear w.r.t. $\mathbf{h}$ and involves several approximations, making it difficult to obtain a closed-form expression for $\mathbf{R}_y$. Therefore, we can only estimate $\hat{\mathbf{R}}_y$ using the available sample measurements. Specifically, consider a set of UE locations $\mathcal{U}$ in the site, with corresponding CSI $\mathcal{H}_\mathcal{U}$ is obtained through ray-tracing simulation or practical collection. The dataset $\mathcal{H}_\mathcal{U}$ contains the site-specific environment information, which the BS serves. Hence, we can maximize the estimated $\hat{\mathbf{R}}_y$ obtained from this sample dataset. 

\subsubsection{Problem Reformulation}
The orthogonality and coverage constraints, (\ref{ortho}) and (\ref{coverage}), are also non-convex w.r.t. $\mathbf{C}$, making them difficult to handle directly. To address these challenges, we apply the penalty method to incorporate these constraints into the objective function. Specifically, let $\hat{\mathbf{C}}$ denote the column-normalized version of $\mathbf{C}$. The penalty for the orthogonality constraint (\ref{ortho}) is defined as
\begin{equation} \label{Lorth}
	{L_{{\rm{orth}}}} \triangleq \left\| {{{{\bf{\widehat C}}}^H}{\bf{\widehat C}} - {{\bf{I}}_K}} \right\|^2,
\end{equation}
which penalizes the off-diagonal entries of the Gram matrix of the probing codebook. If ${L_{{\rm{orth}}}}$ becomes zero, $\mathbf{G}_\mathbf{C}$ is equal to the identity matrix, ensuring orthogonality. \\
\indent In addition, the maximum operator in (\ref{coverage}) is non-differentiable, making it unsuitable for gradient-based methods. To overcome this, we introduce the log-sum-exp operator as a smooth approximation to the maximum beam response. The expression is given by
\begin{equation}
	 \tilde{y}_\text{max} (\mathbf{y}(\mathbf{C})) = \frac{1}{\beta} \log \sum\limits_{k = 1}^K {{e^{\beta {y_k}}}},
\end{equation}
where $\beta \in \mathbb{R}$ is a tuning factor. Consequently, the corresponding coverage constraint is given by
\begin{equation}
	{L_{{\mathop{\rm cov}} }} \triangleq [{y_{{\rm{th}}}} - \tilde{y}_\text{max} (\mathbf{y}(\mathbf{C}))]_+.
\end{equation}
\begin{lemma} \normalfont \label{lemma1}
	Suppose that (i) $\mathcal{U}$ contains an infinite number of i.i.d. sampled UE locations, i.e., $\left| \mathcal{U} \right| \to \infty$, and satisfies $\left\| {\bf{y}} \right\| < \infty$ for each UE; (ii) (P3) admits a KKT point within the feasible region; (iii) $\beta \to \infty$. Then, the following problem (P4) is asymptotically equivalent to (P3) as $\lambda_1, \lambda_2 \to \infty$.
	\begin{equation} \label{P4}
		\left( {\rm{P4}}\right)  \quad  \mathop {\max }\limits_{\mathbf{C} \in \mathcal{C}^{N_t \times K}} \quad \log \det {\bf{\widehat R}}_y^{\left( N \right)} - {\lambda _1}{L_{{\rm{orth}}}} - {\lambda _2}{L_{{\mathop{\rm cov}} }}.
	\end{equation}
\end{lemma}
\begin{IEEEproof}
	Condition (i) ensures that ${\bf{\widehat R}}_y^{\left( \left| \mathcal{U} \right| \right)}\left( {\bf{C}} \right)\mathop  \to \limits^{a.s.} {{\bf{R}}_y}\left( {\bf{C}} \right)$ as $\left| \mathcal{U} \right| \to \infty$ \cite{HDP}. Condition (ii) guarantees the feasibility of (P3). A classical inequality holds for $\tilde{y}_\text{max}$:
	\begin{equation}
		\mathop {\max }\limits_k {y_k} \le \frac{1}{\beta }\log \sum\limits_{k = 1}^K {{e^{\beta {y_k}}}}  \le \mathop {\max }\limits_k {y_k} + \frac{{\log K}}{\beta }.
	\end{equation}
	By the squeeze theorem, $\tilde{y}_\text{max}$ converges uniformly to $\mathop {\max }\limits_k {y_k}$ as $\beta \to \infty$, i.e., condition (iii). Thus, under these conditions, (P4) is asymptotically equivalent to (P3) in the sense that it yields solutions that satisfy the constraints and approach the optimal objective value of (P3) as $\lambda_1, \lambda_2 \to \infty$, under mild regularity conditions \cite{penalty}. 
\end{IEEEproof}
\begin{remark}[Understanding of Lemma~\ref{lemma1}] \normalfont
	Lemma~\ref{lemma1} converts the intractable problem (P3) into the smooth penalized formulation (P4). Since the asymptotic conditions cannot be met in practice, we adopt sufficiently large $|\mathcal U|$, $\beta$, $\lambda_1$, and $\lambda_2$ to obtain a good empirical approximation. In this case, $L_{\rm orth}$ and $L_{\rm cov}$ serve as penalty terms that promote near-orthogonality and reliable coverage among probing beams. We further note that the hinge-type penalty $[y_{\rm th}-\tilde y_{\max}]^{+}$ has zero gradient once $\tilde y_{\max}>y_{\rm th}$, which may slow down training and does not encourage an explicit safety margin. Therefore, in implementation we replace it with a margin-maximization surrogate by setting $\tilde{L}_{\rm cov} \triangleq -\tilde y_{\max}$, so that the objective term $-\lambda_2 L_{\rm cov}^{\rm (impl)}$ directly rewards larger $\tilde y_{\max}$. This surrogate is equivalent to minimizing the hinge loss in the violation regime ($\tilde y_{\max}\le y_{\rm th}$) up to an additive constant, and it continues to enlarge the margin when $\tilde y_{\max}>y_{\rm th}$. The achievable margin increase is inherently bounded due to the array response, hence it will not lead to unbounded growth. Overall, the resulting design maximizes the global information volume in $\mathbf y$ while maintaining geometric diversity in the probing space.
\end{remark}
\vspace{-0.1cm}
\subsubsection{Proposed Solution}
Lemma \ref{lemma1} encourages the dataset size, i.e., $\left| \mathcal{U} \right|$ to be as large as possible. However, computing the full-batch covariance and its log-determinant at every iteration is computationally expensive and memory-intensive. Mini-batch stochastic gradient descent (SGD) offers a scalable alternative by using a randomly sampled subset of users to form a stochastic approximation of the objective and its gradient. Under the standard i.i.d. sampling assumption, the mini-batch estimator yields an unbiased or asymptotically unbiased estimate of the gradient of the empirical objective, enabling efficient optimization within the dataset \cite{penalty}. The detailed algorithm for solving (P4) using mini-batch SGD is provided in \textbf{Algorithm~\ref{alg1}} . Since the mini-batches are drawn from the site-specific dataset, the codebook is optimized by the resulting SGD iterations to capture the site-dependent distribution.
\vspace{-0.1cm}
\begin{remark}[Change of variable] \normalfont
	One remaining challenge in (P4) is that the design variable $\mathbf{C}$ resides on a complex torus. Efficient algorithms for searching solutions in such a manifold include projected gradient descent \cite{penalty} and Riemannian conjugate gradient descent \cite{zhao2025tri}, but these methods are computationally expensive and poorly compatible with the adopted SGD. As a result, we optimize the phase matrix of the probing codebook, i.e., $\boldsymbol{\Phi}=\angle \mathbf{C}$, within the space $[0,2\pi )^{N_t\times K}$, transforming the search space from the original manifold to a subspace of Euclidean space.
\end{remark}
\begin{algorithm}[t]
	\setlength{\textfloatsep}{0.cm}
	\setlength{\floatsep}{0.cm}
	\small
	\caption{Mini-batch SGD for Probing Codebook Design}
	\renewcommand{\algorithmicrequire}{\textbf{Input}}
	\renewcommand{\algorithmicensure}{\textbf{Output}}
	\label{alg1}
	\begin{algorithmic}[1]
		\REQUIRE Site-specific dataset $\{\mathcal{U},\mathcal{H}_\mathcal{U}\}$, batch size $B_1$, step $\eta_{\text{pc}}$, penalty weights $\lambda_1,\lambda_2$, smoothing factor $\beta$, iterations $I_{\text{pc}}$
		\ENSURE Optimized probing codebook $\mathbf C$
		\STATE Initialize $\boldsymbol{\Phi}^{(0)}\in [0,2\pi )^{N_t\times K}$;
		\FOR{$i=1$ to $I_{pc}$}
		\STATE Sample a mini-batch of users $\mathcal{B}_i\subset\mathcal{H}_\mathcal{U}$ with $|\mathcal B_i|=B_1$
		\STATE Construct RSRP measurements using developed model or collect practical measurements $\{\mathbf y_j(\boldsymbol{\Phi}^{(i-1)})\}_{j\in\mathcal B_i}$ 
		\STATE Compute batch mean $\bar{\mathbf y}_{\mathcal B_i}\leftarrow\frac{1}{B_1}\sum_{j\in\mathcal B_i}\mathbf y_j$
		\STATE Compute batch covariance $$
		\widehat{\mathbf R}_{y,\mathcal B_i} \leftarrow \frac{1}{B_1-1}\sum_{j\in\mathcal B_i}(\mathbf y_j-\bar{\mathbf y}_{\mathcal B_i})(\mathbf y_j-\bar{\mathbf y}_{\mathcal B_i})^T$$
		\STATE Evaluate the objective $L_{\mathcal B_i}$ by (\ref{P4});
		\STATE Update by batch gradient $$
		\boldsymbol{\Phi}^{(i)}\leftarrow\boldsymbol{\Phi}^{(i-1)}+\eta_{\text{pc}}\nabla_{\boldsymbol{\Phi}}L_{\mathcal B_i};
		$$
		\ENDFOR
		\STATE Recover the codebook from phase by $\mathbf C^{(I_{\text{pc}})}=\exp\{ j\boldsymbol{\Phi}^{(I_{\text{pc}})} \}$
	\end{algorithmic}
\end{algorithm}
\vspace{-0.1cm}
\subsubsection{Convergence Analysis} \label{remark4}
In each iteration, the covariance matrix $\hat{\mathbf R}_{y,\mathcal{B}_i}$ is computed from a mini-batch of samples and is positive semidefinite by construction. By applying a diagonal loading, $\hat{\mathbf R}_{y,\mathcal{B}_i}\succ 0$ is guaranteed for all iterations, which ensures that $\log\det(\hat{\mathbf R}_{y,\mathcal{B}_i})$ is continuously differentiable with Lipschitz continuous gradients with respect to $\boldsymbol{\Phi}$. Moreover, both the orthogonality penalty $L_{\rm orth}$ and the log-sum-exp surrogate $\tilde y_{\max}$ are smooth functions of $\boldsymbol{\Phi}$. As $\boldsymbol{\Phi}$ lies in a compact domain, the overall objective function admits a finite Lipschitz constant and is therefore $L$-smooth.	At each iteration, the gradient is evaluated using independently sampled UEs, yielding an unbiased stochastic gradient estimator with bounded variance that decreases inversely with the mini-batch size. Under these conditions and with a properly chosen step size, standard results in non-convex stochastic optimization guarantee that \textbf{Algorithm~\ref{alg1}} converges to a first-order stationary point of problem~(\ref{P4}) \cite{sgd}, in the sense that
\begin{equation}
	\lim_{I_{\rm pc}\to\infty} \min_{0\le i\le I_{\rm pc}-1} 
	\mathbb E\!\left[\big\| \nabla_{\boldsymbol{\Phi}} F(\boldsymbol{\Phi}^{(i)})\big\|^2\right] = 0,
\end{equation}
where $F(\boldsymbol{\Phi})$ denotes the expected objective function. Moreover, the convergence rate follows the standard sub-linear behavior of mini-batch stochastic gradient methods, with the residual error floor diminishing as the mini-batch size increases.

\section{CFM-based Beamforming} \label{sec4}
Once the probing codebook $\mathbf C$ is determined, the remaining problem reduces to learning an efficient mapping from the probing measurements to a beamformer that maximizes the expected beamforming utility. Formally, for a fixed probing codebook, problem (P1) reduces to
\begin{subequations}\label{p5}
	\begin{align} 
		\left( {\rm{P5}}\right) \quad \quad  \  \mathop {\max }\limits_{f} \quad & u\left( {{\bf{h}}(\mathcal{E}_s),f\left( {{\bf{y}}} \right)} \right)  \\
		\label{5a} \mathrm{s.t.} \ \  & (\ref{measurement}), f\left( {{\bf{y}}} \right) \in \mathcal{S}_\mathbf{w}.
	\end{align}
\end{subequations}
However, problem (P5) remains intractable in practice. Due to the nonlinear and noisy nature of the RSRP measurements, as well as the inherent ambiguity of power-only observations, multiple channel realizations may give rise to similar RSRP vectors. As a result, the mapping from $\mathbf y$ to the optimal beamformer is generally non-injective, and it is unrealistic to directly regress a unique beamformer from $\mathbf y$. Instead, the solution to (P5) is more appropriately characterized by a conditional distribution of high-utility beamformers given $\mathbf y$. From an implementation perspective, directly optimizing over such a functional space under the constraint (\ref{5a}) is also challenging. To address this issue, we parameterize the mapping using a neural network and generate a small set of candidate beamformers conditioned on the RSRP vector. The final beamformer is then selected from this candidate set according to the achieved utility. This leads to the following practical reformulation of (P5):
\begin{subequations}\label{p6}
	\begin{align} 
		\left( {\rm{P6}}\right) \quad \mathop {\max }\limits_{\theta} \quad & \max_{m} \; u\left( {{\bf{h}}(\mathcal{E}_s),\mathbf{w}_m} \right)  \\
		\label{6a} \mathrm{s.t.} \quad  & (\ref{measurement}),\; \left\lbrace \mathbf{w}_m\right\rbrace_{m=1}^M = f_\theta \left( {{\bf{y}}} \right),\;
		\mathbf{w}_m \in \mathcal{S}_\mathbf{w}.
	\end{align}
\end{subequations}
\subsection{Flow Matching Basis}
FM provides a principled framework for learning complex conditional distributions in such a functional space by modeling a continuous transformation from a simple base distribution to the target distribution. Let $\mathbf{z}_0 \sim p_0(\mathbf{z})$ denotes a sample drawn from a tractable source distribution, and let $\mathbf{z}_1 \sim p_1(\mathbf{z})$ denotes a sample from the unknown target data distribution, both defined on $\mathbb{R}^{N_t}$. FM constructs a family of intermediate distributions ${\left( {{p_t}} \right)_{0 \le t \le 1}}$ that smoothly interpolates between $p_0$ at $t=0$ and $p_1$ at $t=1$. At the sample level, this evolution is governed by a deterministic ordinary differential equation (ODE):
\begin{equation}
	\frac{{d{{\bf{z}}_t}}}{{dt}} = {{\bf{v}}}\left( {{{\bf{z}}_t},t} \right),
\end{equation}
where ${{\bf{v}}}\left( {{{\bf{z}}_t},t} \right):{\mathbb{R}^{{{N_t}\times 1}}} \times \left[ {0,1} \right] \to {\mathbb{R}^{{{N_t}\times 1}}}$ is a time-dependent velocity field. This ODE continuous flow that transports samples from the source distribution to the target distribution. In practice, the source distribution is typically chosen as an isotropic Gaussian, i.e., $p_0=\mathcal{N}(\mathbf{0}_{N_t}, \mathbf{I}_{N_t})$. \\
\indent The corresponding probability density $p_t$ evolves according to the continuity equation
\begin{equation}
	\frac{{\partial {p_t}}}{{\partial t}} =  - {\nabla _{\bf{z}}}\left( {{p_t} \cdot {\bf{v}}\left( {{{\bf{z}}_t},t} \right)} \right),
\end{equation}
which ensures conservation of probability mass along the flow. A key construction in FM is to express the marginal distribution $p_t$ as a mixture of sample-conditional probability paths
\begin{equation}
	{p_t}\left( {\bf{z}} \right) = \int {{p_{t\left| 1 \right.}}\left( {{\bf{z}}\left| {{{\bf{z}}_1}} \right.} \right){p_1}\left( {{{\bf{z}}_1}} \right)d{{\bf{z}}_1}},
\end{equation}
where each conditional path ${p_{t\left| 1 \right.}}\left( {{\bf{z}}\left| {{{\bf{z}}_1}} \right.} \right)$ connects the source distribution at $t=0$ to a distribution that collapses to the target sample $\mathbf{z}_1$ at $t=1$. FM learns a parametric approximation ${{\bf{v}}_{\bf{\theta }}}\left( {{{\bf{z}}_t},t} \right)$ to the velocity field by matching it to a reference velocity under the constructed conditional paths. The model parameters ${\boldsymbol{\theta}}$ are optimized by minimizing the mean-squared error \cite{fm0}:
\begin{equation} \label{FM}
	{L_{\text{FM}}}\left( \boldsymbol{\theta}  \right) = {\mathbb{E}_{{{\bf{z}}_t},t}}\left[ {{{\left\| {{{\bf{v}}_{\bf{\theta }}}\left( {{{\bf{z}}_t},t} \right) - {\bf{v}}\left( {{{\bf{z}}_t},t} \right)} \right\|}^2}} \right],
\end{equation}
where the expectation is taken over the induced intermediate distribution.
\vspace{-0.3cm}
\subsection{Conditional Beamforming Model}
In practice, directly optimizing the objective in (\ref{FM}) is rarely feasible, since the true probability flow velocity associated with the marginal path $p_t$ depends on the unknown target distribution and is therefore intractable. Moreover, the unconditional FM formulation learns a single global vector field that transports samples toward a fixed marginal distribution $p_1(\mathbf{w}^\star)$. Such a formulation is insufficient for the beamforming problem considered in this work. \\
\indent As discussed in the previous section, the optimal beamformer is not drawn from a universal distribution, but is instead determined by the underlying channel realization. The probing feedback acquired during the initial access stage, in the form of the RSRP vector $\mathbf{y}$, provides partial and noisy measurements of the propagation environment, thereby constraining the set of feasible channels and, consequently, the set of optimal beamformers. This implies that the quantity of interest is the conditional distribution ${p_1}\left( {{\bf{w}}^\star\left| {\bf{y}} \right.} \right)$, rather than the marginal distribution $p_1(\mathbf{w}^\star)$. \\
\indent To address these challenges, we adopt CFM, in which the probability flow is explicitly conditioned on the observed RSRP vector. Under CFM, the conditional density evolution satisfies
\begin{equation}
	\frac{{\partial {p_t}\left( {{\bf{z}}\left| {\bf{y}} \right.} \right)}}{{\partial t}} =  - {\nabla _{\bf{z}}}\left( {{p_t}\left( {{\bf{z}}\left| {\bf{y}} \right.} \right) \cdot {\bf{v}}\left( {{{\bf{z}}_t},t;{\bf{y}}} \right)} \right).
\end{equation}
where the velocity field is allowed to depend on the conditioning variable $\mathbf{y}$. \\
\indent A key advantage of CFM is that the training objective simplifies significantly by conditioning the flow on a single target example $\mathbf{z}_1$ sampled from the dataset. Specifically, given paired samples $\mathbf{z}_1$ as condition, CFM constructs an explicit interpolation path between paired samples $(\mathbf{z}_0, \mathbf{z}_1)$ and an intermediate time $t \sim \text{Uni}\left[ 0,1\right]$, CFM constructs an explicit interpolation path between the source and target samples. A commonly adopted choice is linear interpolation as follows:
\begin{equation} \label{interpolation}
	{{\bf{z}}_t} = \left( {1 - t} \right){{\bf{z}}_0} + t{{\bf{z}}_1},
\end{equation}
which induces a valid family of intermediate distributions without requiring explicit evaluation of probability densities. Conditioned on the RSRP vector $\mathbf{y}$, the corresponding reference velocity along this path is uniquely determined as
\begin{equation} \label{velocity}
	{\bf{v}}\left( {{{\bf{z}}_t},t;{\bf{y}}} \right) = \frac{{d{{\bf{z}}_t}}}{{dt}} = {{\bf{z}}_1} - {{\bf{z}}_0}.
\end{equation}
The velocity field is then parameterized by a neural network ${{\bf{v}}_{\boldsymbol{\theta }}}\left( {{{\bf{z}}_t},t;\mathbf{y}} \right)$, where different RSRP realizations induce different flow dynamics and, consequently, different conditional distributions over beamformers. The model parameters $\theta$ are learned by minimizing the following mean squared error objective: 
\begin{equation} \label{CFM}
	{L_{\text{CFM}}}\left( \boldsymbol{\theta}  \right) = {\mathbb{E}_{{{\bf{z}}_0},{{\bf{z}}_1},\mathbf{y},t}}\left[ {{{\left\| {{{\bf{v}}_{\boldsymbol{\theta }}}\left( {{{\bf{z}}_t},t;\mathbf{y}} \right) - \left( {{{\bf{z}}_1} - {{\bf{z}}_0}} \right)} \right\|}^2}} \right].
\end{equation}
This formulation learns an implicit approximation to the posterior distribution ${p_1}\left( {{\bf{w}}^\star\left| {\bf{y}} \right.} \right)$. The conditioning variable ${\bf{y}}$ selects a specific vector field from a family of flows, thereby reshaping the transport dynamics to be consistent with the observed propagation signature encoded in the RSRP measurements.
\vspace{-0.3cm}
\subsection{Conditional Flow Matching Architecture}
\vspace{-0.1cm}
After introducing the design principles and training formulation in the previous subsections, we now present the implementation details of the proposed CFM framework. This subsection focuses on the network architecture, the conditioning mechanism, and the overall training structure.
\begin{figure}[tb]
	\centering
	\includegraphics[scale=0.18]{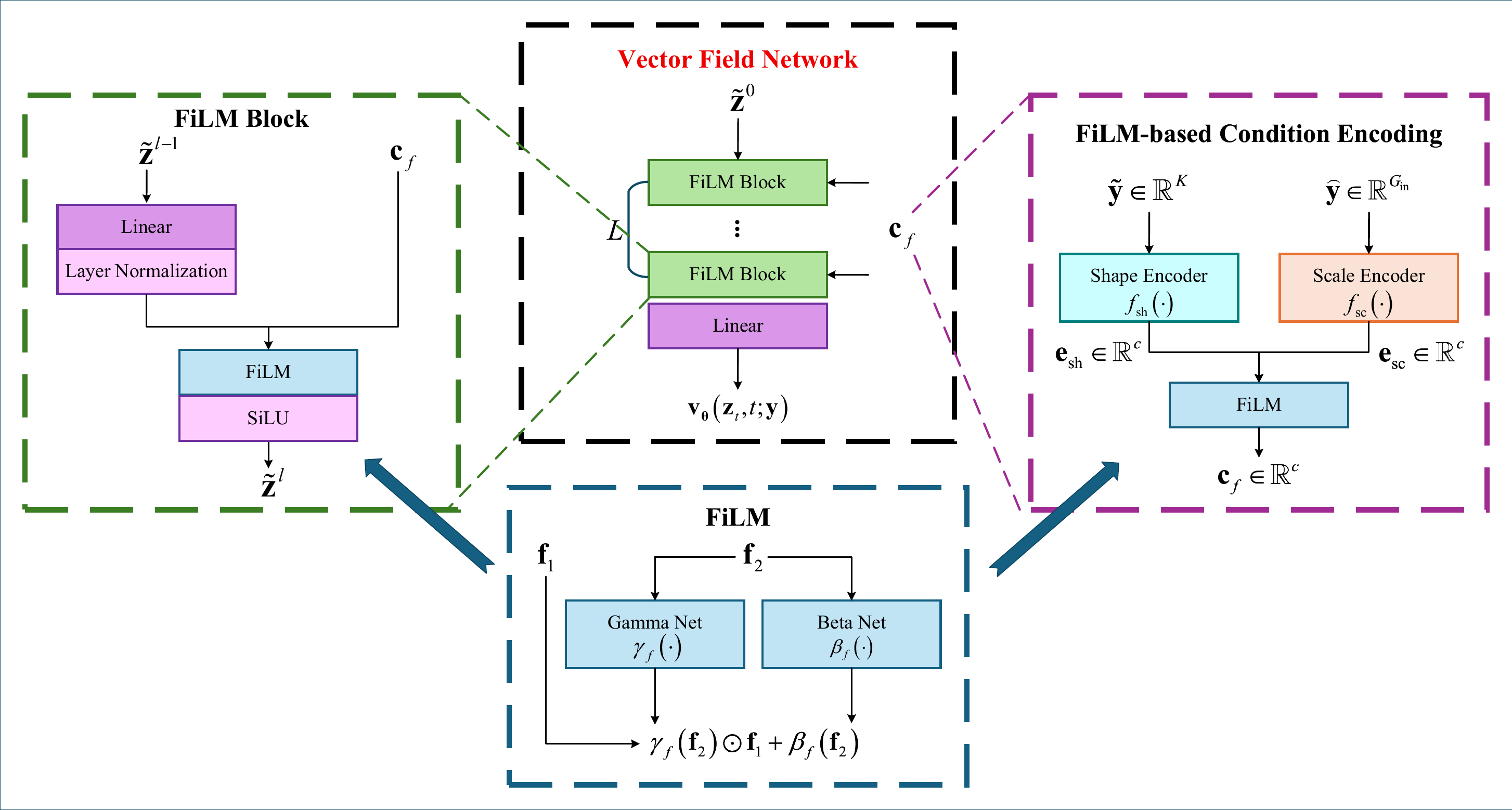}
	\caption{Illustration of the proposed CFM architecture}
	\label{film_cfm}
\end{figure} 
\subsubsection{Network Architecture}
An illustration of the proposed architecture is shown in Fig. \ref{film_cfm}. The CFM model consists of two main components: a vector field network and a condition encoder. The condition encoder adopts a dual-branch structure to process the RSRP vector and extract complementary conditioning features. Specifically, the first branch, referred to as the shape encoder ${f_{{\rm{sh}}}}\left(  \cdot  \right)$, maps the normalized RSRP vector in the logarithmic domain, denoted by $\tilde{\mathbf{y}}$, into a shape embedding $\mathbf{e}_{\text{sh}} \in \mathbb{R}^{c}$. This branch captures the relative beam-wise power pattern and processes the input through a MLP equipped with layer normalization \cite{ba2016layer} and sigmoid Linear Unit (SiLU) activation \cite{ramachandran2017searching}. The resulting embedding primarily encodes geometry-related information associated with the propagation environment. \\
\indent The second branch, referred to as the scale encoder ${f_{{\rm{sc}}}}\left(  \cdot  \right)$, extracts global power-related features from the RSRP vector. Its input consists of $G_{\text{in}}$-dimensional global statistics, denoted by ${{\bf{\mathord{\buildrel{\lower3pt\hbox{$\scriptscriptstyle\frown$}}\over y} }}}$ which summarize the overall scale of the received power. In this work, we adopt the sample mean and standard deviation of the RSRP vector, yielding $G_{\text{in}}=2$. The scale encoder maps these features into a scale embedding $\mathbf{e}_{\text{sc}} \in \mathbb{R}^{c}$ using an architecture similar to that of the shape encoder. The two embeddings are fused using a FiLM module \cite{film} to generate the final conditioning vector, given by
\begin{equation}
	{{\bf{c}}_f} = \gamma_f \left( {{{\bf{e}}_{{\rm{sc}}}}} \right) \odot {{\bf{e}}_{{\rm{sh}}}} + \beta_f \left( {{{\bf{e}}_{{\rm{sc}}}}} \right),
\end{equation}
where $\gamma_f \left( \cdot \right)$ and $\beta_f \left( \cdot \right)$ denote learned affine projections, and $\odot$ represents element-wise multiplication. \\
\indent Based on the resulting conditioning vector, the conditional velocity field ${{\bf{v}}_{\boldsymbol{\theta }}}\left( {{{\bf{z}}_t},t;\mathbf{y}} \right)$ is realized by stacking multiple FiLMBlocks. In the vector field network, the complex-valued state $\mathbf{w}_t^\star=\mathbf{z}_t \in \mathbb{C}^{{N_t}\times 1}$ is first flattened into a real-valued representation $\bar{\mathbf{z}}_t \in \mathbb{R}^{{2N_t}\times 1}$ by concatenating its real and imaginary parts. This vector is then concatenated with the scalar time variable $t$ to form the initial hidden feature ${\bf{\tilde z}}^0=[\bar{\mathbf{z}}_t;t] \in \mathbb{R}^{{(2N_t+1)\times 1}}$. Each FiLMBlock maps ${\bf{\tilde z}}^{l-1} \in \mathbb{R}^{{d_{l-1}}\times 1}$ to ${\bf{\tilde z}}^{l} \in \mathbb{R}^{{d_{l}}\times 1}$. through four successive operations: a linear projection, layer normalization, FiLM-based modulation using the conditioning vector $\mathbf{c}_f$, and a nonlinear activation. Specifically, the transformation in the $l$-th block is given by
\begin{equation}
	{{{\bf{\tilde z}}}^l} = {\rm{SiLU}}\left( {{\gamma _l}\left( {{{\bf{c}}_f}} \right) \odot {\rm{LN}}\left( {{{\bf{W}}_l}{{{\bf{\tilde z}}}^{l - 1}} + {{\bf{b}}_l}} \right) + {\beta _l}\left( {{{\bf{c}}_f}} \right)} \right),
\end{equation}
where ${\rm{LN}}$ denotes layer normalization, ${\rm{SiLU}}$ denotes the SiLU activation function, and $\mathbf{W}_l,\mathbf{b}_l$ are the parameters of the linear layer. The functions ${\gamma _l}$ and ${\beta _l}$ implement FiLM modulation at the $l$-th FiLMBlock.\\
\indent In the proposed architecture, the same conditioning vector $\mathbf{c}_f$ is injected into every FiLMBlock, while each block employs its own set of affine modulation parameters. As a result, the conditioning does not merely act as an auxiliary input, but instead dynamically re-parameterizes the transformations at each layer, effectively steering the geometry of the learned velocity field in a manner consistent with the observed RSRP signature.
\vspace{-0.3cm}
\begin{remark}[Design of condition encoding] \normalfont
	Directly feeding the raw RSRP vector $\mathbf{y}$ into a single network is possible but empirically ineffective, as RSRP values typically lie in a highly negative range and exhibit large dynamic variations. To address this issue, we decompose the RSRP vector into two complementary components: a relative beam-wise power pattern that captures geometry-related information, and a global scale that reflects coarse propagation conditions. By employing separate shape and scale encoders and fusing their outputs through FiLM modulation, the proposed conditioning mechanism disentangles directional and power-related uncertainties, enabling the probability flow to adapt to the partial and noisy nature of RSRP vectors in a structured and stable manner.
\end{remark} 
\vspace{-0.5cm}
\subsection{Training Algorithm and Sampling Method}
\vspace{-0.1cm}
\begin{algorithm}[t]
	\setlength{\textfloatsep}{0.cm}
	\setlength{\floatsep}{0.cm}
	\small
	\renewcommand{\algorithmicrequire}{\textbf{Input}}
	\renewcommand{\algorithmicensure}{\textbf{Output}}
	\caption{Training of FiLM-based CFM}
	\label{alg2}
	\begin{algorithmic}[1]
		\REQUIRE Site-specific dataset $\{\mathcal{U},\mathcal{H}_\mathcal{U}\}$,  codebook $\mathbf{C}$, batch size $B_2$, base standard deviation $\sigma_0$, learning rate $\eta_\theta$, iterations $I_{\theta}$
		\ENSURE Trained parameters $\boldsymbol\theta$ of the vector field $\mathbf v_{\boldsymbol\theta}(\cdot)$
		\STATE Initialize network parameters $\boldsymbol\theta^{(0)}$;
		\FOR{$i=1$ \TO $I_{\theta}$}
		\STATE Sample a mini-batch of users $\mathcal{B}_i\subset\mathcal{H}_\mathcal{U}$ with $|\mathcal B_i|=B_2$;
		\STATE Compute target beamformer $\mathbf{w}_{\mathcal B_i}^{\star}$;
		\STATE Compute conditional RSRP vector $\mathbf{y}$;
		\STATE Sample $\mathbf{t} \in \mathbb{R}^{B_2}$ from ${\rm Uni}[0,1]$;
		\STATE Sample base noise $\mathbf z_0$ from the source Gaussian $\mathcal N(0, \sigma_0^2)$;
		\STATE Construct interpolation state $\mathbf z_t$ by (\ref{interpolation});
		\STATE Compute reference velocity $\mathbf{v}$ by (\ref{velocity});
		\STATE Predict velocity by $\hat{\mathbf v} \leftarrow \mathbf v_{\boldsymbol\theta}\!\big(\mathbf z_t, \mathbf{t} \mid \mathbf{y}\big)$;
		\STATE Compute loss $L_{\text{CFM}}(\boldsymbol\theta) \leftarrow \mathbb{E}_{\mathcal{B}_i}\|\hat{\mathbf v}-\mathbf v\|^2$;
		\STATE Update $\boldsymbol\theta \leftarrow \boldsymbol\theta - \eta_\theta \nabla_{\boldsymbol\theta} L_{\text{CFM}}(\boldsymbol\theta)$;
		\ENDFOR
	\end{algorithmic}
\end{algorithm}
The training procedure of the proposed CFM model is summarized in \textbf{Algorithm~\ref{alg2}.} In this work, the RSRP vector is first processed to extract both relative and global features. Specifically, given the RSRP vector $\mathbf{y}$ in the logarithmic domain, we compute its sample mean and standard deviation as
\begin{equation}
	\mu_s = \frac{1}{K}\sum_{k=1}^{K} y_k, \  \sigma_s = \sqrt{\frac{1}{K}\sum_{k=1}^{K}\!\left(y_k-\mu_s\right)^2}.
\end{equation}
The normalized RSRP vector and the associated global statistics are then given by
\begin{equation}
		\tilde{\mathbf{y}} = \frac{\mathbf{y}-\mu_s\mathbf{ 1}_K}{\sigma_s} \in \mathbb{R}^{K\times 1}, \
		{{\bf{\mathord{\buildrel{\lower3pt\hbox{$\scriptscriptstyle\frown$}}\over y} }}}=[\mu_s, \sigma_s]^T \in \mathbb{R}^{2\times 1},
\end{equation}
which serve as inputs to the shape and scale encoders, respectively. The CFM model is also trained using mini-batch SGD, hence \textbf{Algorithm~\ref{alg2}} follows a convergence analysis similar to Subsubsection~\ref{remark4}. \\
\indent After training, the conditional flow provides a direct mechanism for sampling beamformers from the learned posterior distribution $p\left(\mathbf{w}^\star \mid \mathbf{y} \right)$. At inference time, given an RSRP vector $\mathbf{y}$, the BS draws $M$ independent samples by integrating the conditional flow from multiple source initializations. These samples serve as candidate beamformers that account for the residual uncertainty after probing. To select a beamformer for data transmission, the BS transmits these candidate beams to the UE in the second interaction stage. The UE evaluates the instantaneous array gain associated with each candidate and feeds back the index of the strongest beam. The selected beamformer is then used for data transmission in the subsequent beam-locking stage. The detailed procedure for sampling and determining the beamformer is described in \textbf{Algorithm~\ref{alg3}}.  \\
\indent This selection strategy effectively approximates a maximum a posteriori (MAP) solution, as choosing the beamformer with the highest realized gain corresponds to selecting the most probable and most effective mode under the learned posterior distribution. As the number of samples $M$ increases, the probability that the selected beamformer approaches the true MAP solution improves, while the associated feedback and computational overhead remain moderate. Consequently, a fundamental trade-off arises between beamforming optimality and feedback overhead, governed by the choice of $M$. This trade-off will be quantitatively investigated in the following section.

\begin{algorithm}[t]
	\setlength{\textfloatsep}{0.cm}
	\setlength{\floatsep}{0.cm}
	\small
	\renewcommand{\algorithmicrequire}{\textbf{Input}}
	\renewcommand{\algorithmicensure}{\textbf{Output}}
	\caption{Candidate beamformers sampling and final beamformer determination}
	\label{alg3}
	\begin{algorithmic}[1]
		\REQUIRE Conditioning RSRP input $\mathbf{y}$, trained $\mathbf v_{\boldsymbol\theta}$, base standard variance $\sigma_0$, number of steps $N_{\text{step}}$, temperature $T$, number of samples $M$
		\ENSURE Selected beam $\mathbf{w}^\star$
		\STATE Set $\Delta t \leftarrow 1/N_{\text{step}}$;
		\FOR{$m=1$ \TO $M$}
		\STATE Initialize $\mathbf z^{(m)} \leftarrow T\sigma_0 \boldsymbol\epsilon,\;\boldsymbol\epsilon\sim\mathcal N(\mathbf 0,\mathbf I)$;
		\ENDFOR
		\FOR{$n=0$ \TO $N_{\text{step}}-1$}
		\STATE $t_s \leftarrow n/N_{\text{step}}$;
		\FOR{$m=1$ \TO $M$}
		\STATE $\mathbf u \leftarrow \mathbf v_{\boldsymbol\theta}\!\big(\mathbf z^{(m)}, t_s \mid \mathbf{y}\big)$;
		\STATE $\mathbf z^{(m)} \leftarrow \mathbf z^{(m)} + \Delta t\, \mathbf u$;
		\ENDFOR
		\ENDFOR
		\STATE Direct $M$ sampled candidate beams $\{\mathbf z^{(m)}\}_{m=1}^{M}$ to the UE 
		\STATE UE reports the beam index $m^\star$ with the largest gain;
		\STATE $\mathbf{w}^\star \leftarrow \mathbf{z}^{(m^\star)}$;
	\end{algorithmic}
\end{algorithm}
\vspace{-0.3cm}
\section{Simulation Results} \label{sec5}
In this section, we conduct simulations to evaluate the performance of the proposed framework, including both the probing codebook design and the CFM-based beamforming model.
\begin{figure}[t]
	\centering
	\subfloat[``O1\_28'']{
		\label{o128}
		\includegraphics[scale=0.025]{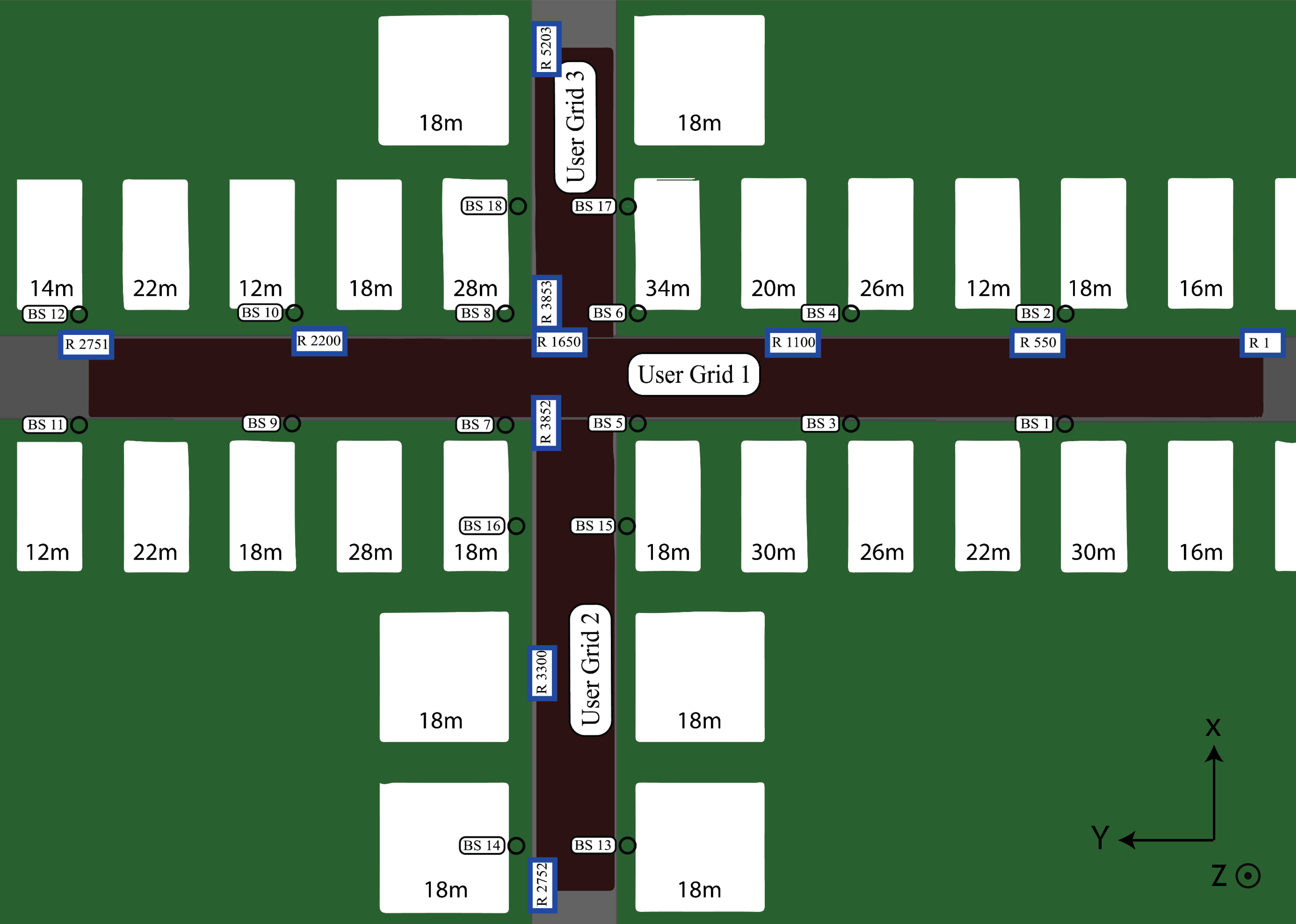}
	} \hfill
	\subfloat[``O1B\_28'']{
		\label{o1b28}
		\includegraphics[scale=0.128]{O1B_28}
	}
	\caption{Illustration of adopted scenarios}
	\label{scenarios}
\end{figure}
\vspace{-0.3cm}
\subsection{Simulation Setup}
Two representative scenarios, denoted as ``O1\_28'' and ``O1B\_28'', from the DeepMIMO dataset \cite{deepmimo} are adopted, as illustrated in Fig.~\ref{scenarios}. For the ``O1\_28'' scenario, we select BS3 and User Grid 1, where the majority of UEs experience line-of-sight (LoS) propagation conditions. This scenario serves as a relatively simple benchmark to verify the basic feasibility and performance of the proposed framework. For the ``O1B\_28'' scenario, the same BS and user grid are considered. However, additional blockages and nearby reflectors are present around the BS, resulting in heterogeneous propagation conditions that include both LoS and non-line-of-sight (NLoS) links. This scenario poses a more challenging environment and is therefore used to assess the robustness of the proposed framework under complex propagation conditions. \\
\indent Channel realizations for both scenarios are generated using the DeepMIMO dataset. The key system parameters and the adopted neural network configurations are summarized in Table~\ref{tab1} and Table~\ref{tab2}, respectively, unless otherwise stated. Since the ray-tracing engine in DeepMIMO already provides realistic channel realizations, a relatively small log-normal shadowing variance is assumed in the simulations, set to $\sigma_{\text{sh}}^2=1$ dB.
\begin{table}[t]
	\small
	\centering
	\caption{Simulation settings}
	\label{tab1}
	\begin{tabular}{|c|c|c|}
		\hline
		\textbf{Parameter}       & \textbf{Description}                           & \textbf{Value}         \\ \hline
		$f_c$                        & Carrier frequency                            & 28 GHz                      \\ \hline
		$\text{BW}$                        & Bandwidth                            & 100 MHz                      \\ \hline
		$P_t$                        & Transmit power                            & 40 dBm                     \\ \hline
		$K$                        & Size of codebook                            & 16                      \\ \hline
		$M$                        & Number of candidate beams                            & 8                      \\ \hline
		$N_t$                        & Number of antennas                             & 64                      \\ \hline
		$L_s$                        & Number of SSB symbols                 & 5                     \\ \hline
		$d$                   & Antenna spacing                 & $\lambda/2$               \\ \hline
		$S_n$                        & Noise power spectrum density & -170 dBm/Hz                      \\ \hline
		$\sigma_\text{sh}^2$                        & Log-variance of the shadowing  & 1 dB                      \\ \hline
		$\lambda_1$                   & Orthogonality penalty                            & 0.1                  \\ \hline
		$\lambda_2$                    & Coverage penalty                             & 0.01                    \\ \hline
		$\beta$ & Smoothing factor                      & 5                 \\ \hline
		$\sigma_0$                 & Base standard deviation                      & 1                  \\ \hline
		$N_{\text{step}}$                     & Number of sampling steps                              & 40                  \\ \hline
		$T$                  & Sampling temperature                          & 0.5 \\ \hline
	\end{tabular}
\end{table}
\begin{table}[bt]
	\centering
	\caption{Architectures and parameters adopted in the simulation}
	\label{tab2}
	\begin{tabular}{cccc}
		\toprule
		\textbf{Layer}        & \textbf{Parameter}      & \textbf{Layer}            & \textbf{Parameter}         \\
		\midrule
		\multicolumn{2}{c}{\textbf{Shape Encoder}}      & \multicolumn{2}{c}{\textbf{Scale Encoder}}             \\
		Linear                & ($K$, 128)                & Linear                    & ($G_\text{in}$, 128)               \\
		LN+SiLU               & 128                     & SiLU                      & /                          \\
		Linear                & (128, 128)              & Linear                    & (128, 128)                 \\
		LN+SiLU               & 128                     & SiLU                      & /                          \\
		\midrule
		\multicolumn{2}{c}{\textbf{FiLM Block}}         & \multicolumn{2}{c}{\textbf{Vector Field Network}}      \\
		Linear                & (129, 256)              & 3 $\times$ FiLM Block              & /                          \\
		LN+SiLU               & 256                     & Linear                    & (256, 128)                 \\
		\midrule
		\multicolumn{2}{c}{\textbf{FiLM in FiLM Block}} & \multicolumn{2}{c}{\textbf{FiLM in Condition Encoder}} \\
		Linear                & (128, 256)              & Linear                    & (128,128)                 \\
		\bottomrule
	\end{tabular}
\end{table}
\vspace{-0.3cm}
\subsection{Evaluation of Proposed SIM Codebook}
We unfold this section by demonstrating the convergence of the proposed codebook design method, the achieved objective, and the properties of the learned codebook. The popular DFT codebook is used as a baseline.  \\
\indent Fig.~\ref{codebook_loss} shows the convergence of the proposed \textbf{Algorithm~\ref{alg1}} under scenario ``O1\_28''. The total training and validation loss curves are plotted in the outer graph, while the three components of the loss, i.e., $\log \det {\bf{\widehat R}}_y^{\left( N \right)}$, ${L_{{\rm{orth}}}}$, are shown in the sub-figures. The algorithm achieves rapid convergence within tens of iterations, particularly for the information and correlation components. The coverage loss exhibits slower convergence due to its logarithmic scale, but eventually converges as well. \\
\begin{figure}[t]
	\centering
	\includegraphics[scale=0.45]{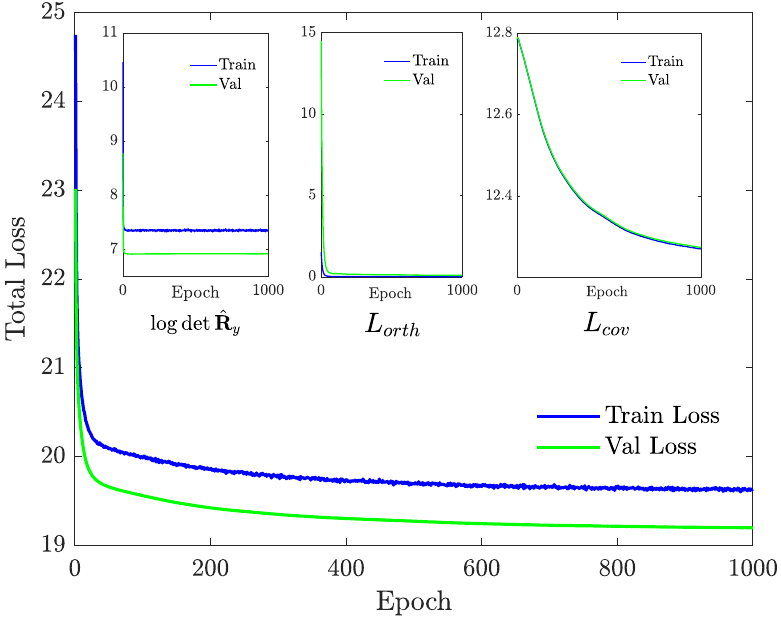}
	\caption{Convergence of the proposed codebook design method}
	\label{codebook_loss}
\end{figure}
\indent Table~\ref{tab3} reports the log-determinant values of the RSRP vectors covariance matrix $\hat{\mathbf{R}}_y$ induced by different probing codebooks, providing the characterization of their information structure. As shown in Table~\ref{tab3}, the DFT codebook exhibits a rapid decrease in log-determinant as the codebook size increases, with a more severe degradation in the blockage-dominated ``O1B\_28'' scenario. This behavior reflects the fact that many DFT probing beams illuminate overlapping angular regions and thus produce highly correlated RSRP responses.  \\
\indent In contrast, the learned codebook maintains nearly constant log-determinant values across different codebook sizes and propagation scenarios, indicating that the resulting probing measurements are information-rich and exhibit limited redundancy. These results suggest that the learned beams adapt to the site-specific propagation geometry and probe distinct dominant paths or scattering directions, thereby providing complementary RSRP measurements. As a result, each probing beam contributes meaningful and non-redundant information, enabling robust beam management even in challenging blockage environments.
\begin{table}[t]
	\small
	\centering
	\caption{Log-determinant value of RSRP covariance matrix obtained by different codebooks}
	\label{tab3}
	\begin{tabular}{cccccc}
		\toprule
		\textbf{O1\_28} & \textbf{DFT} & \textbf{Proposed} & \textbf{O1B\_28} & \textbf{DFT} & \textbf{Proposed} \\
		\midrule
		$K=64$     & -42.81       & -7.02            & $K=64$      & -350.11      & -8.64            \\
		\midrule
		$K=32$     & -38.45       & -6.92            & $K=32$      & -196.36      & -7.20            \\
		\midrule
		$K=16$    & -28.63       & -6.91            & $K=16$      & -89.02       & -7.03            \\
		\midrule
		$K=8$      & -19.82       & -6.90            & $K=8$       & -30.88       & -6.96    \\
		\bottomrule  
	\end{tabular}
\end{table}
\begin{figure*}[t]
	\centering
	\subfloat[``O1\_28'' Euclidean]{
		\label{ed1}
		\includegraphics[scale=0.3]{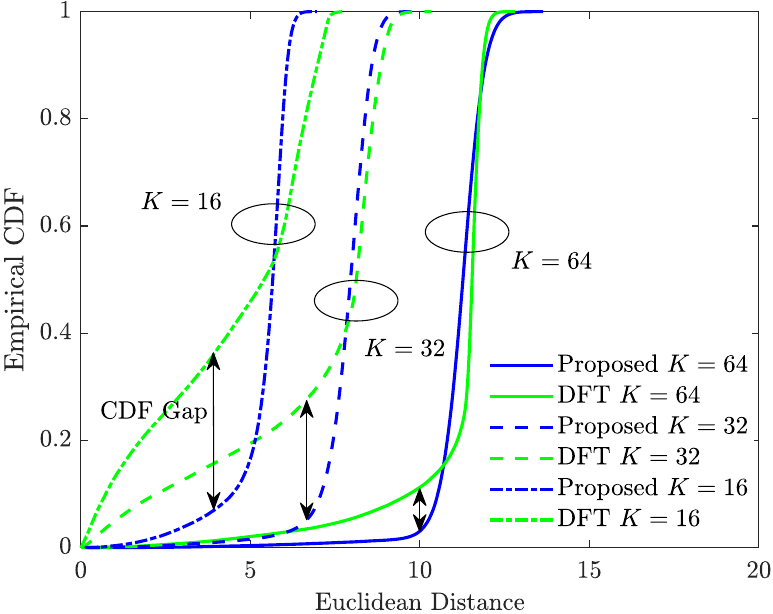}
	} \hfill
	\subfloat[``O1\_28'' Mahalanobis]{
		\label{wd1}
		\includegraphics[scale=0.3]{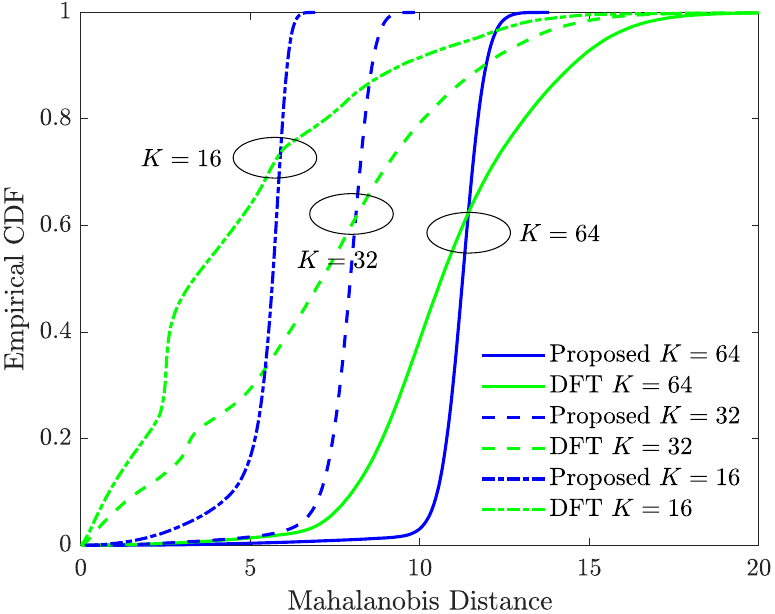}
	} \hfill
	\subfloat[``O1B\_28'' Euclidean]{
		\label{ed2}
		\includegraphics[scale=0.3]{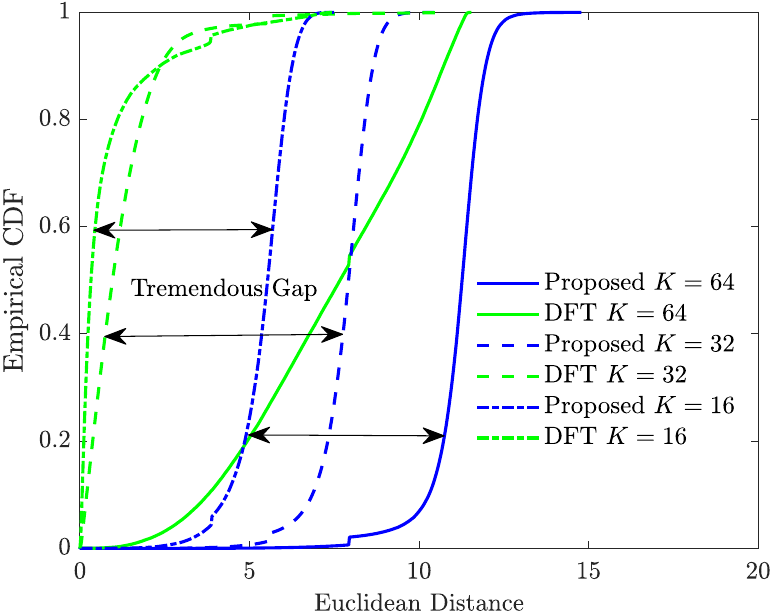}
	} \hfill
	\subfloat[``O1B\_28'' Mahalanobis]{
		\label{wd2}
		\includegraphics[scale=0.3]{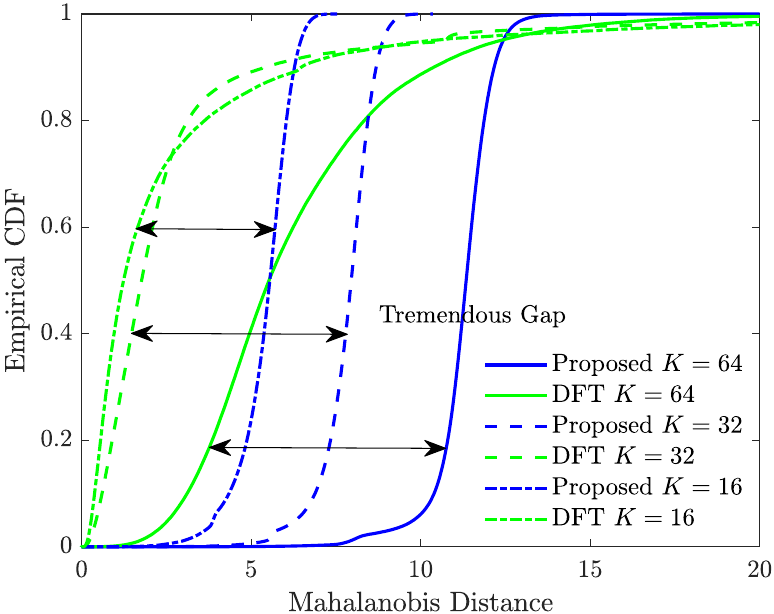}
	}
	\caption{Comparison of RSRP measurements distance between UE pairs}
	\label{dis}
\end{figure*}

Fig.~\ref{dis} shows the empirical cumulative distribution function (CDF) of pairwise Euclidean and Mahalanobis distances between RSRP vectors induced by different probing codebooks. The Euclidean distance reflects the raw dissimilarity between RSRP vectors, while the Mahalanobis distance accounts for their correlation structure and quantifies the effective statistical separability of RSRP measurements. The learned codebook results in distances that concentrate at moderate-to-high values, indicating more uniformly discriminative measurements across UEs. In contrast, the DFT codebook produces a wide-spread distance distribution with a long tail, implying that while a few user pairs are highly separable, many others are weakly distinguishable. This suggests that the learned codebook allocates probing information more evenly across effective dimensions, reducing posterior uncertainty for the majority of UEs, consistent with the objective of maximizing the log-determinant of the RSRP correlation matrix. These findings align with the results from the log-determinant and eigen-spectrum comparisons. \\
\indent Fig. \ref{figcodebook} shows the learned probing codebooks for the two scenarios with $K=8$. In the ``O1\_28'' scenario (first row) the learned beams exhibit structured radiation patterns with two to three pronounced main lobes. These beams focus energy on a small number of angular regions, reflecting the sparse and stable angular support of the underlying channels. In contrast, the ``O1B\_28'' scenario (second row), which includes partial blockages, shows irregular radiation patterns with no dominant lobes in blocked regions. This behavior suggests that, due to blockage, signal propagation is dominated by weak multi-path components, and sharp beams in these regions offer little additional information. Consequently, the learned codebook avoids allocating excessive probing energy to blocked directions and instead focuses on unblocked regions where reliable propagation paths exist. \\
\indent From both a physical and information-theoretic perspective, this behavior is expected. In LoS-dominant environments, probing beams that focus on a few angular directions maximize the sensitivity of RSRP measurements to small variations in user location, enhancing its distinguishability. In contrast, in NLoS conditions, angular energy becomes more dispersed, and sharper beams in blocked regions are less informative. The learned codebooks therefore adapt to the environment's propagation characteristics, balancing angular resolution and robustness. They preserve high-gain lobes in reliable regions and adopt low-gain, irregular patterns in blocked areas, effectively capturing environment-specific channel characteristics.
\begin{figure*}[t]
	\centering
	\begin{minipage}{0.11\linewidth}
		\centering
		\includegraphics[width=\linewidth]{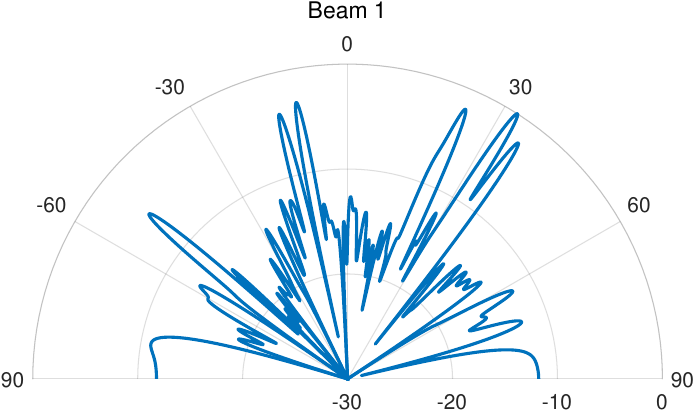}
	\end{minipage}
	\hfill
	\begin{minipage}{0.11\linewidth}
		\centering
		\includegraphics[width=\linewidth]{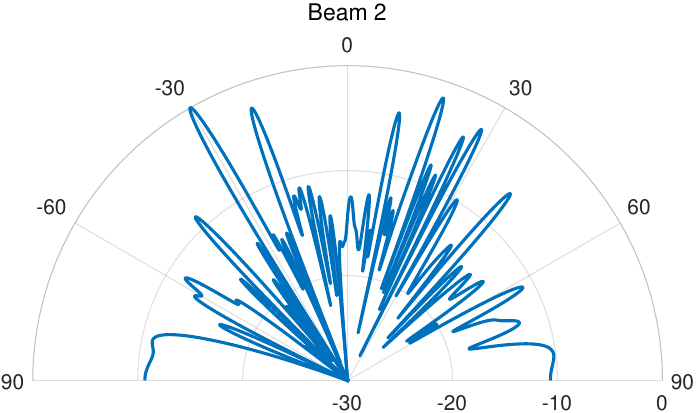}
	\end{minipage}
	\hfill
	\begin{minipage}{0.11\linewidth}
		\centering
		\includegraphics[width=\linewidth]{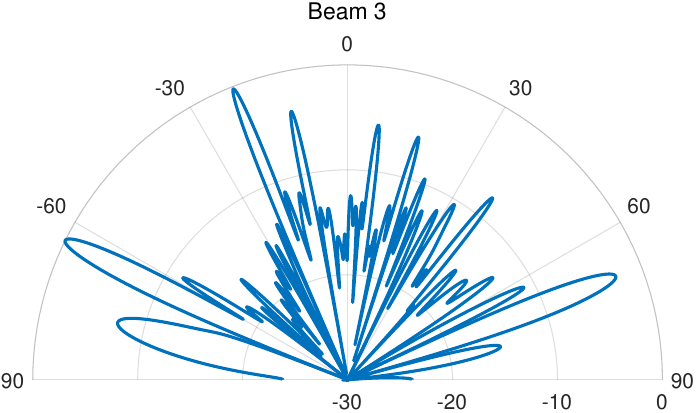}
	\end{minipage}
	\hfill
	\begin{minipage}{0.11\linewidth}
		\centering
		\includegraphics[width=\linewidth]{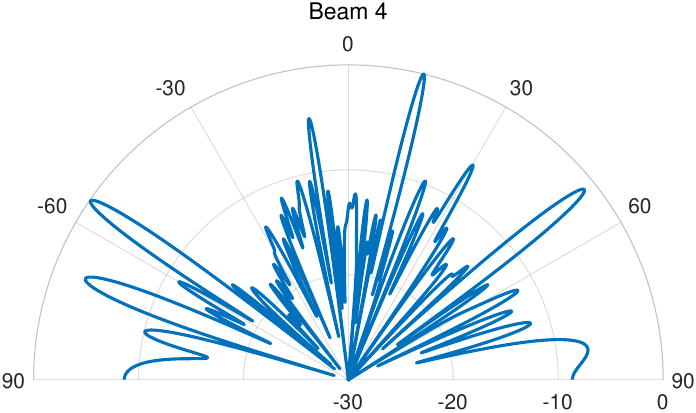}
	\end{minipage}
	\hfill
	\begin{minipage}{0.11\linewidth}
		\centering
		\includegraphics[width=\linewidth]{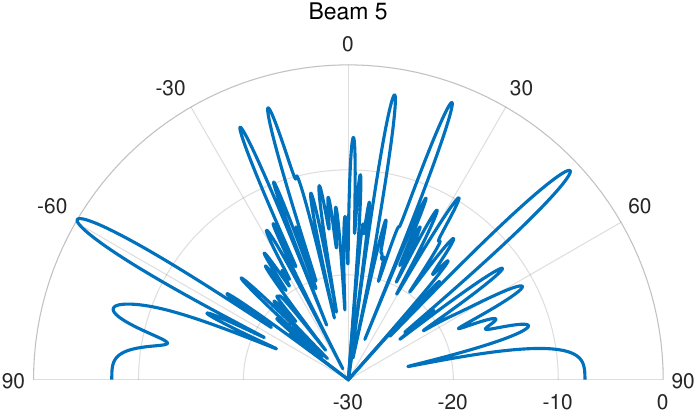}
	\end{minipage}
	\hfill
	\begin{minipage}{0.11\linewidth}
		\centering
		\includegraphics[width=\linewidth]{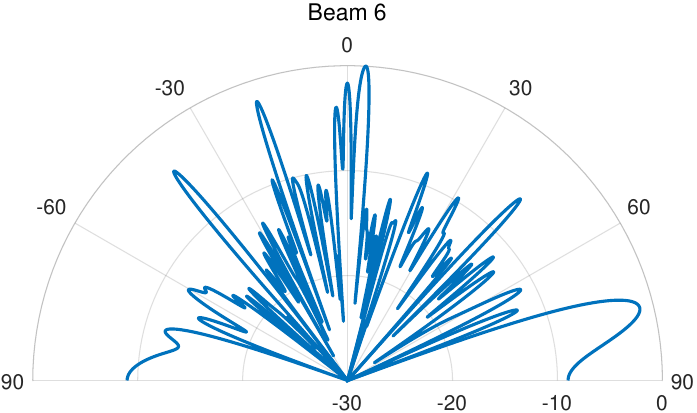}
	\end{minipage}
	\hfill
	\begin{minipage}{0.11\linewidth}
		\centering
		\includegraphics[width=\linewidth]{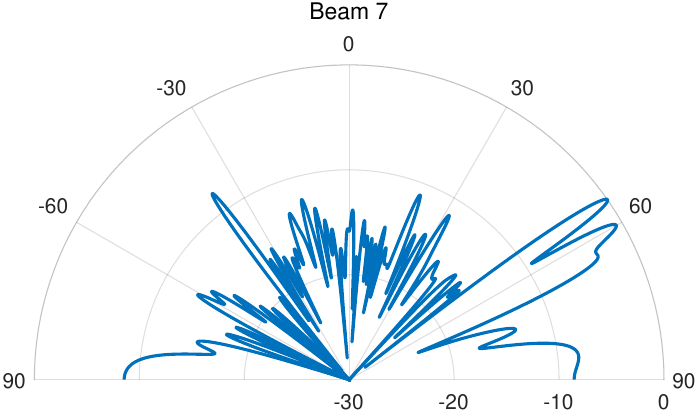}
	\end{minipage}
	\hfill
	\begin{minipage}{0.11\linewidth}
		\centering
		\includegraphics[width=\linewidth]{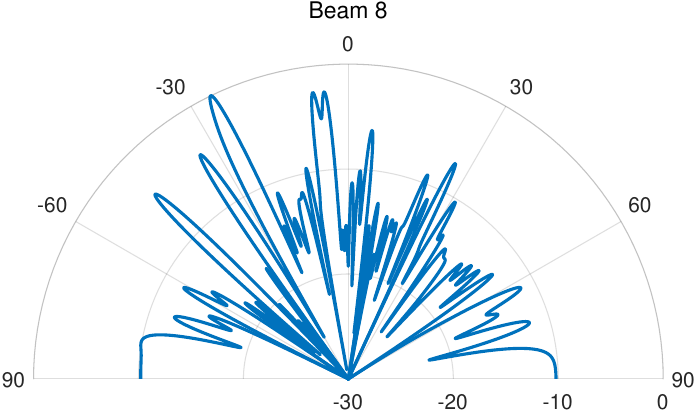}
	\end{minipage} \\
	\begin{minipage}{0.11\linewidth}
		\centering
		\includegraphics[width=\linewidth]{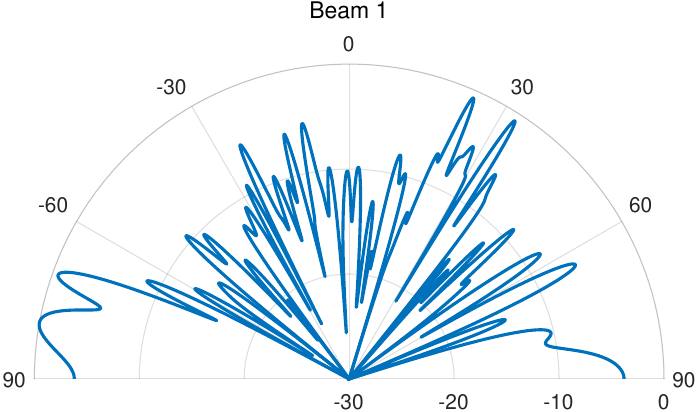}
	\end{minipage}
	\hfill
	\begin{minipage}{0.11\linewidth}
		\centering
		\includegraphics[width=\linewidth]{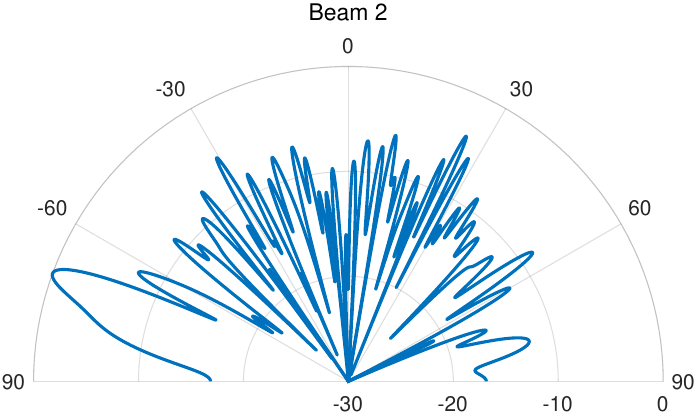}
	\end{minipage}
	\hfill
	\begin{minipage}{0.11\linewidth}
		\centering
		\includegraphics[width=\linewidth]{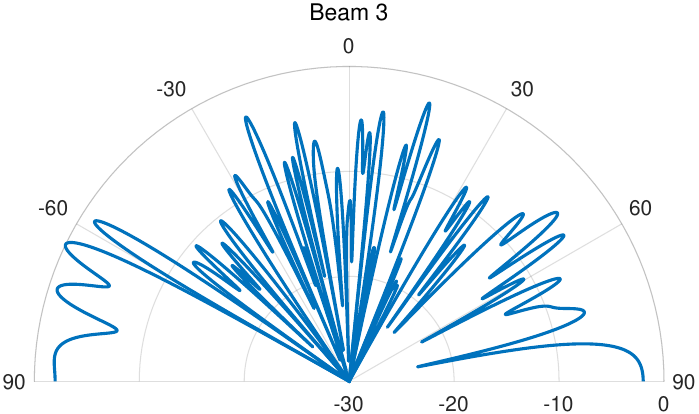}
	\end{minipage}
	\hfill
	\begin{minipage}{0.11\linewidth}
		\centering
		\includegraphics[width=\linewidth]{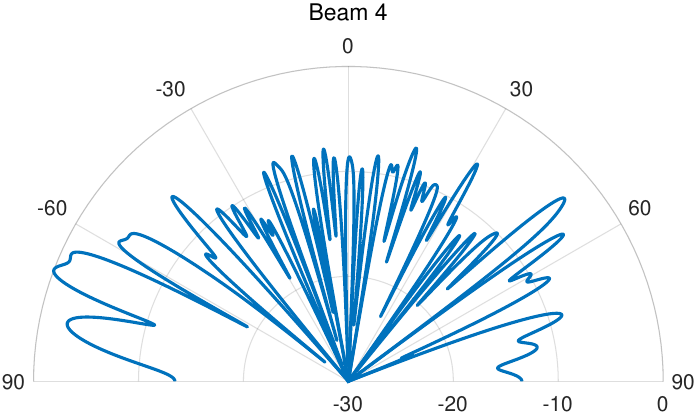}
	\end{minipage}
	\hfill
	\begin{minipage}{0.11\linewidth}
		\centering
		\includegraphics[width=\linewidth]{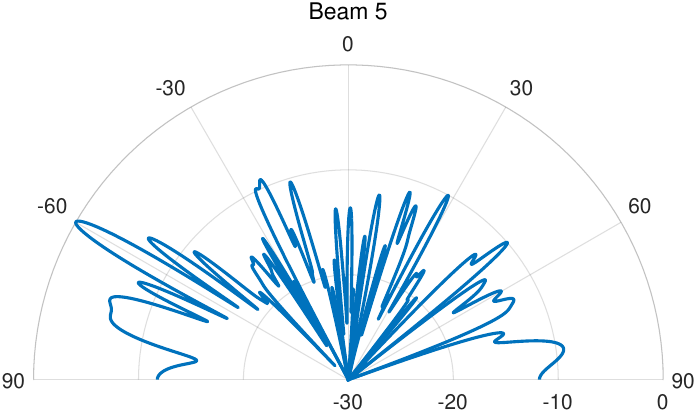}
	\end{minipage}
	\hfill
	\begin{minipage}{0.11\linewidth}
		\centering
		\includegraphics[width=\linewidth]{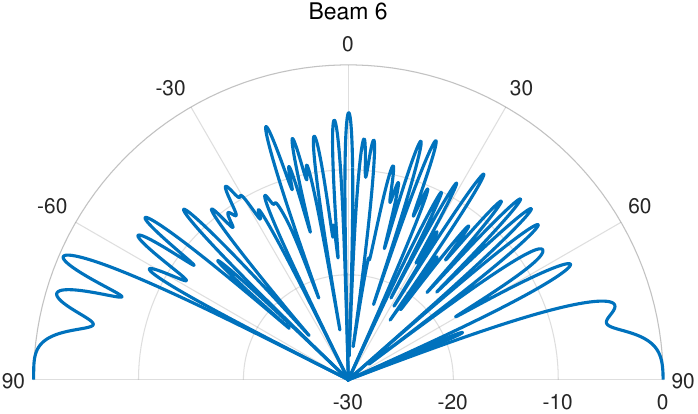}
	\end{minipage}
	\hfill
	\begin{minipage}{0.11\linewidth}
		\centering
		\includegraphics[width=\linewidth]{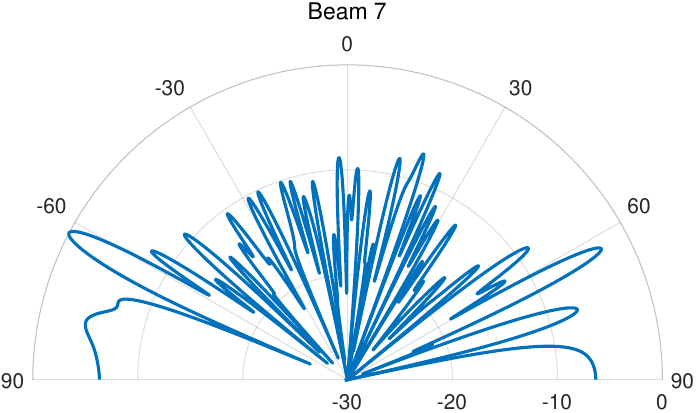}
	\end{minipage}
	\hfill
	\begin{minipage}{0.11\linewidth}
		\centering
		\includegraphics[width=\linewidth]{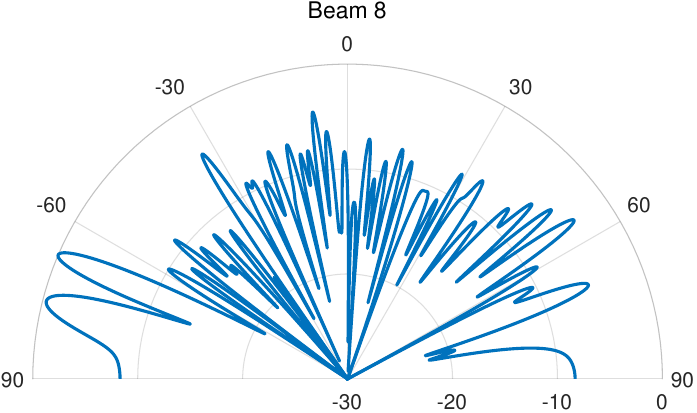}
	\end{minipage}
	\caption{Beam pattern illustrations for the learned codebook (Scenario ``O1\_28'' Beams 1–8 in the first row; Scenario ``O1b\_28'' Beams 1–8 in the second row).}
	\label{figcodebook}
\end{figure*}
\vspace{-0.5cm}
\subsection{Evaluation of Proposed CFM-based Beamforming}
To show the advantage of the proposed CFM model, we adopt three baselines: 1) the existing DisSSBF framework \cite{SSBFCB2, SSBF1}, labeled as ``MLP''; 2) conventional PMI approach, i.e., 2-tier hierarchical beam search using DFT and over-sampled DFT codebooks \cite{PMI}; 3) conventional SRS approach, i.e., generating beamformer via linear minimum mean squarer error (LMMSE) channel estimation. We assume a idealized performance for the LMMSE, i.e., its normalized mean square error (NMSE) achieves its lower bound as $\text{NMSE}=1/(1+L_s\text{SNR})$ \cite{LMMSE}. The estimated channel is then simulated as the ground truth with Gaussian perturbation governed by this NMSE. Furthermore, an implicit baseline is the MRT since the adopted models are evaluated by the array gain gap to the MRT.
\begin{figure}[t]
	\centering
	\includegraphics[scale=0.45]{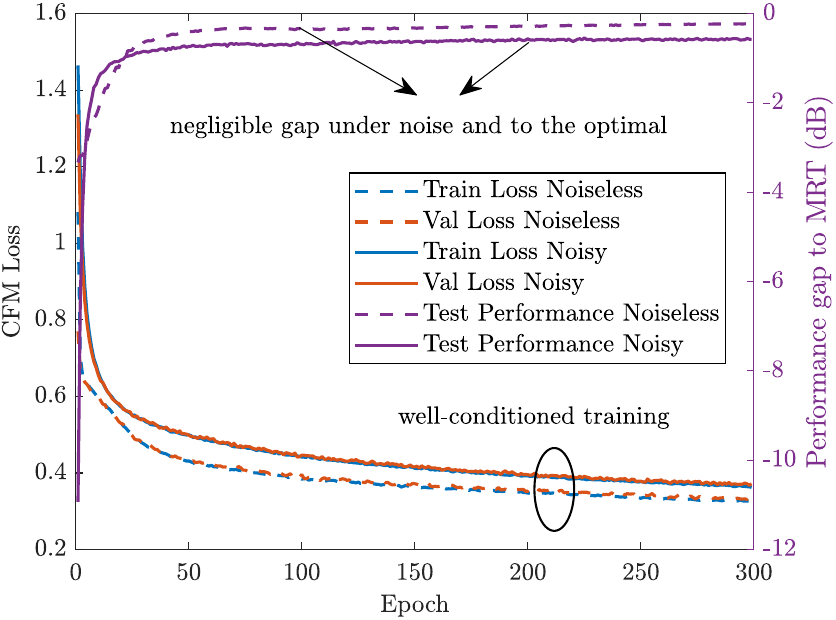}
	\caption{Convergence of the proposed CFM model}
	\label{cfm_loss}
\end{figure}

The convergence of the proposed CFM model, i.e., \textbf{Algorithms~\ref{alg2}} and\textbf{~\ref{alg3}} under scenario ``O1\_28'' is presented in Fig.~\ref{cfm_loss}. Both noiseless and noisy cases are considered. The noise and shadowing setting can refer to Table~\ref{tab1}. The purple curves corresponding to the right y-axis record the best beamforming gain among 8 candidates compared with the MRT. Noiselessly, the model shows a rapid convergence within 100 epochs. The performance gap between the best beam and the MRT is negligible and as small as 0.2 dB. However, even with shadowing and noise perturbation, the proposed CFM still exhibits a similar convergence. The performance gap to MRT only suffers a small degradation to about 0.5 dB and remains negligible. This result confirms both the effectiveness and robustness of the proposed CFM model and stage II in the framework. 
\begin{figure}[t]
	\centering
	\includegraphics[scale=0.45]{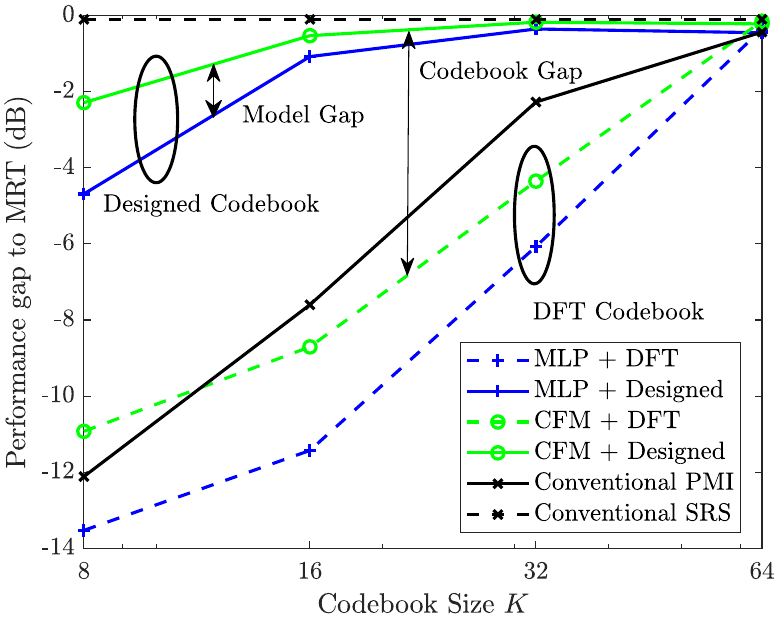}
	\caption{Performance comparison with different codebook size}
	\label{cfm_cdsize}
\end{figure}

Performance comparison for different approaches versus the codebook size is shown in Fig.~\ref{cfm_cdsize}. ``O1\_28'' scenario is adopted for this comparison. First, for both the MLP and CFM, using the designed codebook achieves significantly better performance than using the DFT codebook. Second, the proposed CFM model always outperforms the discriminative model, which confirms the previous analysis about the uncertain nature of beamformer or CSI given RSRP. Moreover, the proposed CFM also outperforms the conventional PMI approach due to the flexibility of the proposed model. The proposed CFM also keeps a small gain with the conventional SRS approach, which requires additional overhead and is granted an idealized setting. This gap vanishes as the codebook enlarges.

\begin{figure}[t]
	\centering
	\includegraphics[scale=0.45]{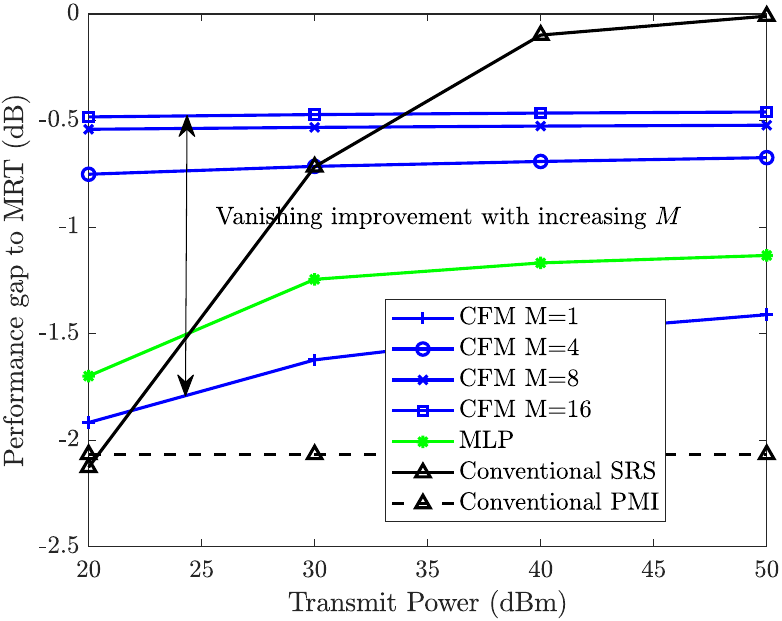}
	\caption{Performance gap to MRT versus transmit power in ``O1\_28''}
	\label{cfm_ptx_o128}
\end{figure}
\begin{figure}[t]
	\centering
	\includegraphics[scale=0.45]{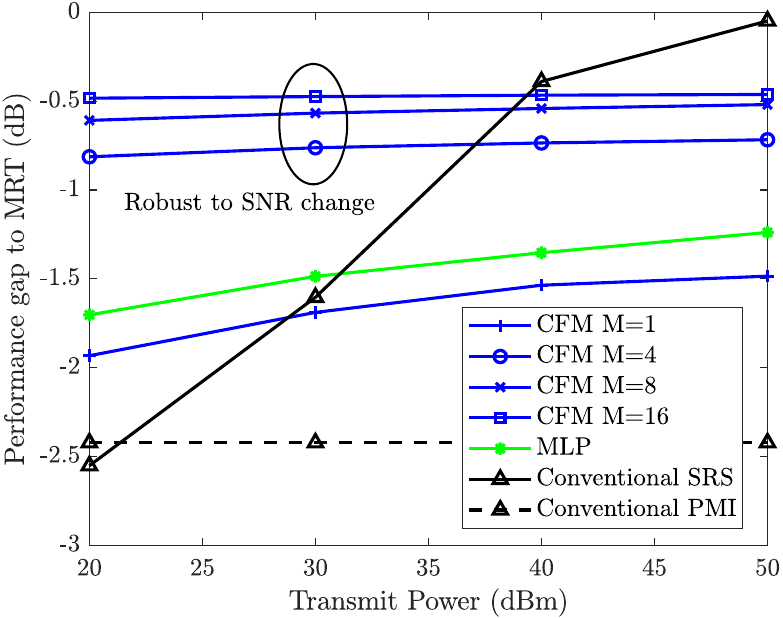}
	\caption{Performance gap to MRT versus transmit power in ``O1B\_28''}
	\label{cfm_ptx_o1b28}
\end{figure}

Fig.~\ref{cfm_ptx_o128} and~\ref{cfm_ptx_o1b28} illustrate the variation of performance gap to the optimal of different methods as the change of transmit power, i.e., SNR. Since the receiving SNR of a specific UE varies across the serving region depending on the distance to BS, we observe the average performance of UEs in the dataset by fixing the noise setting and changing transmit power. $K=32$ and $M=16$ is set for the conventional PMI scheme. Firstly, in terms of the candidate beam number, it is observed that the performance of the proposed CFM gradually approaches to the optimal as the number of candidate beams increases. However, this improvement gradually vanishes with increasing $M$. In terms of the transmit power, it is shown that, unlike the conventional CSI-based beamforming scheme, which is susceptible to the SNR change, the proposed CFM model is robust to this change with a moderate number of candidate beams.
\vspace{-0.3cm}
\subsection{Discussion}
Summarizing the results in Fig.~\ref{cfm_cdsize},~\ref{cfm_ptx_o128}, and~\ref{cfm_ptx_o1b28}, the advantages of the proposed GenSSBF is validated. Firstly, compared with the conventional PMI scheme, the proposed framework can achieve better of similar performance with significantly lower beam sweeping overhead. For the conventional SRS scheme, although facing a slight performance gap in the high SNR regime, the proposed framework exhibits stronger robustness to the decrease of SNR. Moreover, as long as a proper set of candidate beams are delivered in the beam refinement stage, the proposed framework always outperforms the existing DisSSBF framework. This comparison is summarized in Table~\ref{tab5}. Near-optimal and sub-optimal are abbreviated as n-optimal and s-optimal, respectively. 
\begin{table}[t]
	\small
	\centering
	\caption{Comparison of different beam management frameworks}
	\label{tab5}
	\begin{tabular}{ccccc}
		\toprule
		Framework & GenSSBF & DisSSBF & SRS & PMI \\
		\midrule
		Overhead     & \textbf{very low}       & \textbf{very low}            & moderate      & low            \\
		\midrule
		Performance     & \textbf{n-optimal}       & s-optimal            & \textbf{n-optimal}      & s-optimal            \\
		\bottomrule  
	\end{tabular}
\end{table}
\vspace{-0.2cm}
\section{Conclusion} \label{sec6}
This paper has proposed a novel beam management framework. Compared with the existing codebook-based (PMI) and CSI-based (SRS) schemes, the proposed framework utilized the generative artificial intelligence model, CFM in this paper, to capture the site-specific environment information, hence called GenSSBF, achieving near-optimal beamforming with very low overhead. Specifically, the framework first performs coarse channel probing using the optimized SIM codebook, which is designed based on site-specific data to maximize the amount of channel information conveyed. Then, the CFM model exploits the feedback RSRP measurements to generate a set of candidate beams, among which the strongest beam achieves near-optimal performance and is selected for data transmission. Extensive simulation results demonstrate the effectiveness of the proposed GenSSBF framework, as well as the advantages of the proposed codebook design principles and the CFM model.
\vspace{-0.2cm}
\balance
\bibliographystyle{IEEEtran}
\bibliography{reference/mybib}

\end{document}